\documentclass{amsart}
\usepackage[margin=2.2cm]{geometry}

\renewcommand{\paragraph}[1]    {\medbreak\noindent {\bf #1}}

\usepackage{amssymb,amsmath,amsfonts}
\usepackage{graphicx}
\usepackage{amsthm}
\usepackage[title]{appendix}
\usepackage{xcolor}
\usepackage{manyfoot}
\usepackage{booktabs}
\usepackage{algorithm}
\usepackage{algorithmicx}
\usepackage{algpseudocode}
\usepackage[skip=12pt]{caption}
\usepackage{doi}
\usepackage{subfig}

\usepackage[capitalize,noabbrev]{cleveref}

\renewcommand{\url}[1]{\href{https://#1}{#1}}

\usepackage{enumitem} \setlist{noitemsep,nolistsep}
\usepackage{units}

\usepackage{fnpct}

\newtheorem{theorem}{Theorem}
\newtheorem{lemma}[theorem]{Lemma}
\newtheorem{corollary}[theorem]{Corollary}
\newtheorem{definition}[theorem]{Definition}

\newtheorem*{remark*}{Remark}

\newcommand{\loss}{{\mathcal{L}}}
\newcommand{\ssx}{\sigma}
\newcommand{\tsx}{\tau}
\newcommand{\col}[2]{#1[\cdot,#2]}
\newcommand{\row}[2]{#1[#2,\cdot]}
\newcommand{\low}{\operatorname{low}}

\newcommand{\Hgr}{{\sf H}}
\newcommand{\ot}{\leftarrow}

\newcommand{\RR}{\mathbb{R}}
\newcommand{\Rsp}{\RR}
\newcommand{\Xsp}{\mathbb{X}}
\newcommand{\eps}{\varepsilon}
\newcommand{\ee}{\eps}

\newcommand{\Dgm}{\operatorname{Dgm}}
\newcommand{\bigO}{\operatorname{O}}

\renewcommand{\Vert}{\operatorname{Vert}}
\newcommand{\target}{\operatorname{target}}

\newcommand{\manM}{\mathcal{M}}
\newcommand{\manN}{\mathcal{N}}
\newcommand{\matrA}{\mathbf{A}}
\newcommand{\matrB}{\mathbf{B}}
\newcommand{\matrC}{\mathbf{C}}
\newcommand{\matrX}{\mathbf{X}}

\algtext*{EndWhile}
\algtext*{EndIf}%
\algtext*{EndFor}%

\title{Topological Optimization with Big Steps}

\author{Arnur Nigmetov, Dmitriy Morozov\\ \textit{Lawrence Berkeley National Laboratory, 1 Cyclotron Road, Berkeley, CA 94704, USA}}

\begin{document}

\maketitle

\begin{abstract}
    Using persistent homology to guide optimization has emerged as a novel application of topological data
    analysis. Existing methods treat persistence calculation as a black box and
    backpropagate gradients only onto the simplices involved in particular
    pairs. We show how the cycles and chains used in the persistence calculation
    can be used to prescribe gradients to larger subsets of the domain.
    In particular, we show that in a special case, which serves as a building
    block for general losses, the problem can be solved exactly in linear time.
    This relies on another contribution of this paper, which
    eliminates the need to examine a factorial number of permutations of
    simplices with the same value.
    We present empirical experiments that show the practical
    benefits of our algorithm: the number of steps required for the
    optimization is reduced by an order of magnitude.
\end{abstract}

\section{Introduction}
\label{sec:introduction}
Topological optimization~\cite{Gameiro2016,Poulenard2018,chen2019}
is a novel application of persistent homology~\cite{ELZ02}.
The basic idea is to define a loss in terms of the points of a
persistence diagram and minimize it using the modern optimization software that
combines automatic differentiation with state-of-the-art optimization techniques.
Depending on the application, we may want to
reduce noisy features in the data by moving low-persistence points closer to the
diagonal~\cite{Poulenard2018,persistence-sensitive-optimization} or outside of a
particular quadrant~\cite{chen2019}, amplify signal by moving high-persistence points away from
the diagonal~\cite{gabrielsson2020}, match a template signal by moving the current diagram towards a
prescribed one~\cite{Gameiro2016,leygonie2021gradient}, among other applications.
Most of the work so far has been motivated by problems in machine learning: a
loss formulated via persistence can be used to regularize a decision boundary
(following the philosophy that overfitting produces a topologically complex
surface).

Optimization also offers a new approach to an old
problem.  Given a function $f: \Xsp \to \Rsp$ on some topological space,
persistence-sensitive simplification~\cite{simplification_2manifolds} asks
for a nearby function $g: \Xsp \to \Rsp$, with the same persistence diagram as $f$,
but without the points closer than $\ee$ to the diagonal. The original
paper~\cite{simplification_2manifolds} showed that one can solve the problem
for extrema --- and therefore, by duality, completely on 2-manifolds --- but the
suggested algorithm was ad hoc. The running time was later improved to
linear~\cite{attali2009persistence,Bauer2012}, see also~\cite{Tierny2012}.
Crucially, the problem has only been solved for extrema, with the difficulty of
processing middle dimensions (e.g., simplifying 1-dimensional persistence for
functions on 3-manifolds) highlighted by the connection to the Poincar\'e
conjecture~\cite[Section 3.5]{Morozov2008}: because simplification is
impossible for sufficiently large values of $\ee$ on homology spheres, any such
scheme must take the topology of the domain into account.

But it is possible to take a ``best effort'' approach. Instead of solving
the problem exactly via combinatorics, one can formulate a simplification loss
that penalizes points closer to the diagonal than $\ee$. Minimizing such a loss
may not produce the perfect solution $g$, but it can get very close. Moreover,
it offers the flexibility of articulating more sophisticated goals:
for example, to not only simplify the function overall, but also to control the
topology of its specific levelsets or sublevel sets.

\paragraph{Approach.}
We are interested in the general problem, where the loss is formulated as a
partial matching. Some of the points $p_i$ in the persistence diagram are
prescribed targets $q_i$, and the loss aims to minimize the distance between
them, e.g., $\loss = \sum_i (p_i - q_i)^2$. The existing approaches to
this optimization are all based on the same idea. Each point in the
diagram is defined by the values of a pair of simplices:
$p_i = (b_i,d_i) = (f(\ssx_i), f(\tsx_i))$, where $f: K \to \Rsp$ is the input
filtration. The gradient $\partial \loss / \partial p_i$ defined by the loss
immediately translates to the gradient on the simplex values,
$\partial \loss / \partial f(\ssx_i)$ and $\partial \loss / f(\tsx_i)$. These in
turn can be backpropagated through the filtration to define
the gradients on the input data.

This approach is general --- it can handle arbitrary losses --- and comes with
theoretical guarantees of convergence~\cite{Carriere2021}.  But it is also slow.
Subsampling~\cite{solomon2021fast} has been suggested as a way to speed it up.
We instead improve performance by focusing on one of its shortcomings, namely that
it treats persistence as a black box.
The only information used comes from the pairing,
which means that only critical values of the input get any gradient information:
each point $p_i$ gives gradients on only two simplex values.
Moreover, if the optimization is done carefully and multiple simplices get the
same value, the above gradient definition is not even correct:
defining it exactly in general requires examining $k!$ different orders of the
$k$ simplices with the same value~\cite{leygonie2021gradient}.

At the same time, persistent homology computes a lot more structure than just
the pairing represented in the persistence diagram. The standard
algorithms~\cite{ELZ02,dualities} compute cycles and chains in the
domain that certify the existence of a particular pair. We take advantage of
this extra information to speed up optimization by suggesting a principled way
for each point to define gradients for a large set of simplex values.

\paragraph{}%
Our work has four main {\bf contributions}:
\begin{enumerate}[noitemsep,nolistsep]
    \item
        We show that for a simple loss, called \emph{singleton loss}, defined by
        matching a single point in the persistence diagram to a target, the
        gradient can be computed exactly, including when multiple simplices have the
        same value, by examining a single permutation, rather than $k!$ required
        in general.
        This structural realization leads to a cubic algorithm to optimize the
        singleton loss.
    \item
        We show that this algorithm can be improved to linear time by examining
        matrices computed as a byproduct of finding the persistence pairing.
    \item
        We introduce a set of heuristics for combining ``big steps'' prescribed
        by individual points into a gradient on both critical and regular
        simplices, which can be backpropagated and optimized using standard
        algorithms and software.
    \item
        We show experimentally that our procedure requires an order
        of magnitude fewer steps to optimize a loss than the standard procedure
        that defines the gradient on only two simplices per persistence pair.
\end{enumerate}

\section{Background}
\label{sec:background}

We assume the reader's familiarity with algebraic topology and only briefly
review the setting of persistent homology, to establish the notation. We refer
the reader to the extensive resources~\cite{comp_top_book,ph-survey} for a thorough introduction.

\paragraph{Persistent homology.}
Given a simplicial complex $K$, with $n$ simplices, and a function $f: K \to
\Rsp$ that respects the face relation --- i.e., $f(\ssx) \leq f(\tsx)$ if $\ssx$
is a face of $\tsx$ --- we sort the simplices in $K$ by function value, breaking
ties if necessary so that faces come before their cofaces. We use $<$ to denote
the resulting total order on the simplices. We denote the subcomplexes defined
by the prefixes of this order with $K_i$.
Their nested sequence is called
a \emph{filtration}:
\[
    K_1 \subseteq K_2 \subseteq \ldots \subseteq K_n = K.
\]

Using coefficients in a field and passing to homology, we get a sequence of
homology groups, connected by linear maps induced by the inclusions:
\[
    \Hgr_*(K_1) \to \Hgr_*(K_2) \to \ldots \to \Hgr_*(K_n).
\]
Persistent homology tracks how classes appear and disappear in this sequence, and
produces a set of pairs $(\ssx_i,\ssx_j)$ such that a homology class
created by simplex $\ssx_i$ dies when simplex $\ssx_j$ enters the filtration,
and a set of infinite pairs $(\ssx_i,\infty)$, if a class created by simplex $\ssx_i$
does not die.

To compute this pairing, we start with the boundary matrices, $D_p$, of the
simplicial complex, whose columns and rows are ordered by the filtration.
Each such matrix stores the boundaries of the $p$-simplices\footnote{Recall that a $p$-simplex has $(p+1)$ vertices.}.
It will be convenient to use the simplices themselves to index the columns of
various matrices, so for example $\col{D_p}{\tsx}$ refers to the column
that stores the boundary of $p$-simplex $\tsx$; similarly, $\row{D_p}{\ssx}$
refers to the row that stores the coboundary of $(p-1)$-simplex $\ssx$.

Persistence pairing is computed by reducing the boundary matrix, which can be
interpreted~\cite{vineyards} as finding decompositions $R_p = D_pV_p$, where
matrices $R_p$ are reduced, meaning the lowest non-zeros in their columns appear in
unique rows, and matrices $V_p$ are invertible upper-triangular. There are many such
decompositions --- \cref{alg:ELZ} is the original algorithm~\cite{ELZ02} that
finds one of them --- but the locations of the lowest non-zeros in
matrices $R_p$ are unique and give the persistence pairing. Denoting by
$\low \col{R_p}{\tsx}$ the simplex that corresponds to the row of the lowest non-zero entry in the
column, we have a pair $(\ssx, \tsx)$ iff $\low \col{R_p}{\tsx} = \ssx$
and a pair $(\ssx, \infty)$ iff $\col{R_p}{\ssx} = 0$ and there is no column
with $\low \col{R_p}{\tsx} = \ssx$.
We call such $\ssx$ \emph{positive} or \emph{birth} simplices, and such $\tsx$
\emph{negative} or \emph{death} simplices.

As in \cite{vineyards}, we denote by $U_p$ the inverse of matrix $V_p$, so that
$D_p = R_pU_p$. \cref{alg:ELZ} shows how to compute matrices $R_p,V_p,U_p$. The columns of
$R_p$ and $V_p$ have a natural interpretation: matrix $R_p$ stores the cycles that
generate the homology classes in the respective subcomplexes. Matrix
$V_p$ stores the chains that turn those cycles into boundaries.

\begin{remark*}
    Columns of matrices $D_p$ and $R_p$ are indexed by the $p$-simplices; their
    rows, by the $(p-1)$-simplices. Both rows and columns of matrices $V_p$ and
    $U_p$ are indexed by the $p$-simplices.
\end{remark*}

\begin{algorithm}
\caption{Lazy reduction of the boundary matrix.}
\label{alg:ELZ}
\begin{algorithmic}[1]
    \State $R_p = D_p, V_p = I, U_p = I$ for all $p$ \;
    \ForAll{$\tsx_j \in K$ \textit{(in filtration order)}}
        \While{$\col{R_p}{\tsx_j} \neq 0$ and $\exists~\tsx_i < \tsx_j, \low \col{R_p}{\tsx_i} = \low \col{R_p}{\tsx_j}$}
            \State $\ssx = \low \col{R_p}{\tsx_j}$ \;
            \State $\alpha = R_p[\ssx, \tsx_j] / R_p[\ssx, \tsx_i]$ \;
            \State $\col{R_p}{\tsx_j} = \col{R_p}{\tsx_j} - \alpha \cdot \col{R_p}{\tsx_i}$ \;
            \State $\col{V_p}{\tsx_j} = \col{V_p}{\tsx_j} - \alpha \cdot \col{V_p}{\tsx_i}$ \;
            \State $\row{U_p}{\tsx_i} = \row{U_p}{\tsx_i} + \alpha \cdot \row{U_p}{\tsx_j}$
            \State \qquad (equivalently, $U_p[\tsx_i, \tsx_j] = \alpha$) \;
        \EndWhile
    \EndFor
\end{algorithmic}
\end{algorithm}

Matrices $U_p$ and $V_p$ obtained via the lazy reduction in \cref{alg:ELZ}
have a special property that we rely on below.
Throughout the paper --- starting from the statement and proof of the following
lemma --- it is convenient to simplify the language by assuming that if $\col{R_p}{\tsx} = 0$,
then $\low \col{R_p}{\tsx}$ is implicitly equal to a ``dummy'' simplex
$\bar{\ssx}$ that precedes every other simplex in the filtration order.

\begin{lemma}[Lazy reduction]
    \label{lem:lazy-reduction}
    If decompositions $R_p = D_pV_p$ and $D_p = R_pU_p$ are obtained via the lazy reduction
    in \cref{alg:ELZ},
    then if $\ssx_i = \low \col{R_p}{\tsx_i}$ and $\ssx_j = \low \col{R_p}{\tsx_j}$
    are such that $\tsx_i < \tsx_j$ and $\ssx_i < \ssx_j$,
    then $U_p[\tsx_i,\tsx_j] = V_p[\tsx_i,\tsx_j] = 0$.
\end{lemma}
\begin{proof}
    The proof is by induction. The statement is trivially true initially, when
    $V_p = U_p = I$. Suppose the statement is true after $l-1$ steps of the
    reduction. Suppose in step $l$ we are adding a multiple of column
    $\col{R_p}{\tsx_i}$ to $\col{R_p}{\tsx_j}$. Since the reduction is lazy, it
    means $\low \col{R_p}{\tsx_i} = \low \col{R_p}{\tsx_j}$ before the addition, and
    $\low \col{R_p}{\tsx_j} < \low \col{R_p}{\tsx_i}$ afterwards.
    The corresponding operation in matrix $V_p$ adds a multiple of column
    $\col{V_p}{\tsx_i}$ to column $\col{V_p}{\tsx_j}$, so
    the only non-zero entries that may be introduced into the column
    $\col{V_p}{\tsx_j}$ are those in the column $\col{V_p}{\tsx_i}$.
    By induction all of them fall in rows $\tsx_k$ with
    $\low \col{R_p}{\tsx_k} \geq \low \col{R_p}{\tsx_i} > \low \col{R_p}{\tsx_j}$.
    Since by the time we are reducing $\col{R_p}{\tsx_j}$, we have already reduced
    all the preceding columns --- and therefore their pairs don't change ---
    the claim follows for matrix $V_p$.

    In matrix $U_p$, the corresponding operation is adding a multiple of row
    $\row{U_p}{\tsx_j}$ to row $\row{U_p}{\tsx_i}$. By induction any non-zero in
    the former falls in the columns $\tsx_k$ with
    $\low \col{R_p}{\tsx_k} \leq \low \col{R_p}{\tsx_j} < \low \col{R_p}{\tsx_i}$.
    Since the already reduced columns in $R_p$ don't change,
    the claim follows for matrix $U_p$.
\end{proof}

The following two corollaries follow immediately as contrapositive statements of
the lemma. In both, because matrix $R_p$ is reduced, the equality
among the lowest entries is achieved iff $\tsx_i = \tsx_j$.

\begin{corollary}
    \label{cor:lazy-U}
    If after a lazy reduction entry $U[\tsx_i, \tsx_j] \neq 0$, then
    $\low \col{R_p}{\tsx_j} \leq \low \col{R_p}{\tsx_i}$.
\end{corollary}

\begin{corollary}
    \label{cor:lazy-V}
    If after a lazy reduction entry $V[\tsx_i, \tsx_j] \neq 0$, then
    $\low \col{R_p}{\tsx_j} \leq \low \col{R_p}{\tsx_i}$.
\end{corollary}

\paragraph{Duality.}
Passing from the filtration to cohomology, a vector space dual of homology, we
get a sequence of cohomology groups, connected by linear maps induced by
restrictions:
\[
    \Hgr^*(K_1) \ot \Hgr^*(K_2) \ot \ldots \ot \Hgr^*(K_n).
\]
By duality~\cite{dualities}, the pairing in this sequence is the same as for homology, but with
the role of birth and death reversed, a fact we exploit below.

Algorithmically, we replace the boundary matrix by its anti-transpose,
$D_p^\bot$, i.e., a transpose of $D_p$ with rows and columns ordered in reverse
filtration order. Applying \cref{alg:ELZ}, we get decompositions
$R_p^\bot = D_p^\bot V_p^\bot$ and $D_p^\bot = R_p^\bot U_p^\bot$.
Similar to homology, the matrices have immediate interpretation: $R_p^\bot$ stores
the cocycles and $V_p^\bot$ the cochains that turn them into coboundaries.

\begin{remark*}
    Matrices $R_p^\bot, V_p^\bot, U_p^\bot$ are not anti-transposes of matrices
    $R_p, V_p, U_p$. In $D_p^\bot$ and $R_p^\bot$, rows are indexed by
    $p$-simplices; columns, by $(p-1)$-simplices. In $V_p^\bot$ and $U_p^\bot$,
    both rows and columns are indexed by $(p-1)$-simplices.
\end{remark*}

Persistence pairing is the same for homology and cohomology~\cite{dualities}.
$\low \col{R_p}{\tsx} = \ssx$ iff $\low \col{R_p^\bot}{\ssx} = \tsx$
(simplices $\ssx$ and $\tsx$ are paired).
$\col{R_{p-1}}{\ssx} = 0$ and $\not\exists \tsx$ with $\low \col{R_p}{\tsx} = \ssx$
iff $\col{R_p^\bot}{\ssx} = 0$ and $\not\exists \rho$ with $\low \col{R_{p-1}^\bot}{\rho} = \ssx$
(simplex $\ssx$ is unpaired).

\paragraph{Stability.}
In the combinatorial setting, the following statement is
equivalent~\cite{vineyards} to the stability of persistent homology.

\begin{lemma}
    \label{lem:transposition}
    Suppose two simplices $\ssx_1$ and $\ssx_2$ that appear consecutively in the
    filtration transpose. The only persistence pairs that can change are the two
    pairs that involve simplices $\ssx_1$ and $\ssx_2$.
\end{lemma}

It follows immediately that re-ordering more than two simplices only affects
their respective pairs.

\begin{corollary}
    \label{cor:reordering}
    Given a contiguous set of simplices $X$ in the filtration, changing the
    order of simplices in $X$ can change only persistence pairs with one of the
    endpoints in set $X$.
\end{corollary}

\section{Singleton Loss}
\label{sec:singleton-loss}

Virtually every topological loss proposed in the literature can be rephrased as
a partial matching: some points in the diagram are prescribed targets, where
they need to move. For example, the simplification loss,
\begin{equation}
    \loss_\ee(f) = \sum_{\substack{(b,d) \in \Dgm(f)\\(d - b) \leq \ee}} (d-b)^2
    \label{eq:simplification-loss}
\end{equation}
can be formulated as a partial matching $M$, where every point $(b,d) \in \Dgm(f)$
with $(d-b) \leq \ee$ is matched to the point $((b+d)/2, (b+d)/2)$. Then the
loss can be re-written as
\[
    \loss_\ee(f) = \sum_{(p,q) \in M} (p-q)^2.
\]

We consider the simplest such setting, where the partial matching consists of a
single pair, $(p,q)$. We call this \emph{singleton loss}. We assume that
$p = (b,d) = (f(\ssx), f(\tsx))$ and $q = (b',d')$.
The loss itself is $\loss = (p - q)^2$. We will follow the gradient flow of this
loss and keep track of point $p$. Specifically, we denote by $p_t$ the image of $p$ under
gradient flow after time $t$, and write the loss $\loss_t = (p_t - q)^2$, which
defines the gradient at every point in time.
We note that the loss is oblivious to what happens to the other points in the
persistence diagram.

Following the gradient to minimize this loss translates into moving simplices
$\ssx$ and $\tsx$ in the filtration to their target values $b'$ and $d'$. We
first focus on the negative $p$-simplex $\tsx$. Suppose there are $m$ simplices with values
between $d$ and $d'$, and $m_p$ of them are $p$-simplices.
As we increase or decrease the value of $\tsx$ (depending
on whether $d' > d$ or $d' < d$), it is going to reach the value of each one of
the $m_p$ $p$-simplices. For each such simplex $\tsx_k$, we must determine what happens if
we place it before $\tsx$, when increasing the value, or after $\tsx$, when
decreasing.  If doing so changes the pairing of $\ssx$ to $\tsx_k$, then
$\tsx_k$ needs to move together with $\tsx$ (and enter the \emph{critical set}
$X_\ssx$, defined below). If not, we can safely skip over $\tsx_k$ (a fact that
itself requires a proof).

It is not immediately obvious, but we prove in the next subsection that when
determining the fate of $\tsx_k$, it is not necessary to consider all $k!$
possible orders of the simplices that are moving together with $\tsx$ (as one
might reasonably expect in general~\cite{leygonie2021gradient}); determining the
pairing for a single order suffices.

Besides moving the simplices of the same dimension as $\tsx$, which are the only
simplices that may take over the pairing with $\ssx$, we must move all of their
cofaces, when increasing the value, or their faces, when decreasing the value.
This is required simply to ensure that our simplex order defines a filtration.
We revisit this topic in \cref{sec:faces-cofaces}.

\subsection{Critical Set}
\label{sec:critical-set}

As we move a $p$-simplex
$\tsx$, paired with a $q$-simplex $\ssx$\footnote{The only possible values of $q$ are $(p-1)$ or $(p+1)$.},
in the filtration, we maintain a critical set of $p$-simplices that move together
with $\tsx$ under the gradient flow of the singleton loss.
We say that a set of $p$-simplices is \emph{contiguous}, if their columns are
contiguous in matrix $D_p$.

\begin{definition}
    \label{def:critical-set}
    Given a $q$-simplex $\ssx$, a set of contiguous $p$-simplices $X_\ssx$
    is \emph{critical}, if placing any $\tsx' \in X_\ssx$ as the first simplex
    (when increasing the value of $\tsx$) or as the last simplex (when
    decreasing the value of $\tsx$) in the set makes it paired with $\ssx$.
\end{definition}

The critical set is well-defined because whether simplices $\ssx$ and $\tsx'$
are paired depends only on what simplices appear between $\ssx$ and $\tsx'$, not
on their order. In other words, re-ordering the simplices in the critical set
$X_\ssx$ does not change the pairing of the (first or last) simplex in the
set that is paired with $\ssx$.
This argument implies that when we add a simplex to the critical set, we don't
lose any of the simplices already in it.

\begin{lemma}
    \label{lem:critical-set-add}
    Suppose $X_\ssx^{k-1}$ is a critical set and $\tsx_k$ appears immediately before or
    after the set (depending on the direction that $\tsx$ is moving).
    Suppose that transposing $X_\ssx^{k-1}$ and $\tsx_k$ changes the pairing of
    $\ssx$ to $\tsx_k$. Then $X_\ssx^k = X_\ssx^{k-1} \cup \tsx_k$ becomes the
    critical set after the transposition.
\end{lemma}

A key property of the critical set, expressed in the following lemma, is that it
is resilient under transpositions. If a simplex $\tsx'$ can transpose with the
critical set without becoming paired with $\ssx$, then the critical set does not
change after the transposition.

\begin{lemma}
    \label{lem:critical-set}
    Suppose $X_\ssx^{k-1}$ is a critical set and $\tsx_k$ appears immediately before or
    after the set (depending on the direction that $\tsx$ is moving).
    Suppose that transposing $X_\ssx^{k-1}$ and $\tsx_k$ does not change the pairing of
    $\ssx$. Then $X_\ssx^k = X_\ssx^{k-1}$ remains critical after the transposition.
\end{lemma}
\begin{proof}
    Suppose we are decreasing the value of $\tsx$ and, therefore, by definition,
    any simplex in $X_\ssx^{k-1}$, when placed last, is paired with $\ssx$. Let $\tsx'$
    be this last simplex in $X_\ssx^{k-1}$ in the filtration order. Transpose
    $\tsx_k$ with all but the last simplex in the critical set. The last
    simplex, $\tsx'$, remains paired with $\ssx$ (by \cref{cor:reordering}).
    Now transpose $\tsx_k$ and $\tsx'$.
    Their pairing doesn't change, since $\tsx'$ does not become paired with
    $\ssx$ by the assumption of the lemma, and $\tsx'$ remains paired with
    $\ssx$. Since this argument holds for every $\tsx' \in X_\ssx^{k-1}$, the
    critical set does not change.

    The same argument applies when increasing the value of $\tsx$ by replacing ``last''
    with ``first.''
\end{proof}

\cref{lem:critical-set-add,lem:critical-set} together mean that as we move
simplex $\tsx$, the critical set can only grow: simplices enter, but never
leave.
\cref{lem:critical-set} suggests \cref{alg:transpositions}
for changing the value of $\tsx$, using transpositions~\cite{vineyards}:
for each of the $m_p$ $p$-simplices
with values between $d$ and $d'$, transpose it past the critical set. If
its pairing changes to $\ssx$ (or if the transposition is impossible because it
is a face or a coface of one of the simplices in the critical set), add it to
the critical set. Because each transposition takes linear time~\cite{vineyards},
the first for-loop runs in $\bigO(m^2 n)$ time.
The second for-loop can be implemented as a breadth-first search through the
graph of the face--coface relationships
(called a Hasse diagram), so it takes $\bigO(dm)$ time, where $d$ is the dimension of $K$.
Because $d < n$, the former dominates, and we get $\bigO(m^2 n)$ running time
for the whole algorithm.

\begin{algorithm}
    \caption{Moving $\tsx$ using individual transpositions.}
    \begin{algorithmic}[1]
    \State $X_\ssx^1 = \{ \tsx \}$ \;
    \For{\textit{each $p$-simplex $\tsx_k$ with $f(\tsx_k)$ from $d$ to $d'$}}
        \State {transpose $\tsx_k$ with each simplex in $X_\ssx^{k-1}$,}
        \State {\quad updating the pairing using the algorithm in \cite{vineyards}}
        \label{lin:transpose-forward} \;
        \If{$\tsx_k$ \textit{becomes paired with} $\ssx$}
            \State $X_\ssx^k = X_\ssx^{k-1} \cup \{ \tsx_k \}$ \;
            \State {transpose $\tsx_k$ with each simplex in $X_\ssx^{k-1}$,}
            \State {\quad undoing the transpositions in \cref*{lin:transpose-forward}, }
            \State {\quad returning it to the opposite end of $X_\ssx^k$ }
            \label{lin:transpose-back}
        \Else
            \State $X_\ssx^k = X_\ssx^{k-1}$ \;
        \EndIf
    \EndFor
    \For{\textit{each} $\tsx \in X_\ssx^{m_p}$}
        \label{lin:face-coface-for-begin}
        \State set $f(\tsx) = d'$ (ties broken implicitly via the original order) \;
        \If{$d' > d$}
            \State {// move cofaces}
            \For{\textit{each simplex $\rho \supseteq \tsx$, and $d < f(\rho) < d'$}}
                \State set $f(\rho) = d'$ \;
            \EndFor
        \Else
            \State {// move faces}
            \For{\textit{each simplex $\ssx \subseteq \tsx$, and $d' < f(\ssx) < d$}}
                \State set $f(\ssx) = d'$ \;
                \label{lin:face-coface-for-end}
            \EndFor
        \EndIf
    \EndFor
    \end{algorithmic}
    \label{alg:transpositions}
\end{algorithm}

\begin{remark*}
    The transpositions in \cref*{lin:transpose-back} are unnecessary, but they
    simplify the proofs below.
\end{remark*}

Our main contribution is an algorithm for identifying the entire critical set in
$\bigO(m)$ time, without having to perform the transpositions. 
The resulting effect is illustrated in \cref{fig:big-step}, where the gradient
flow implicitly traced by \cref{alg:transpositions} follows the brown
curve. By identifying the critical set, we can move directly to the final
destination --- taking a ``big step'' --- as illustrated with the blue curve.

\begin{figure}
    \centering
    \includegraphics{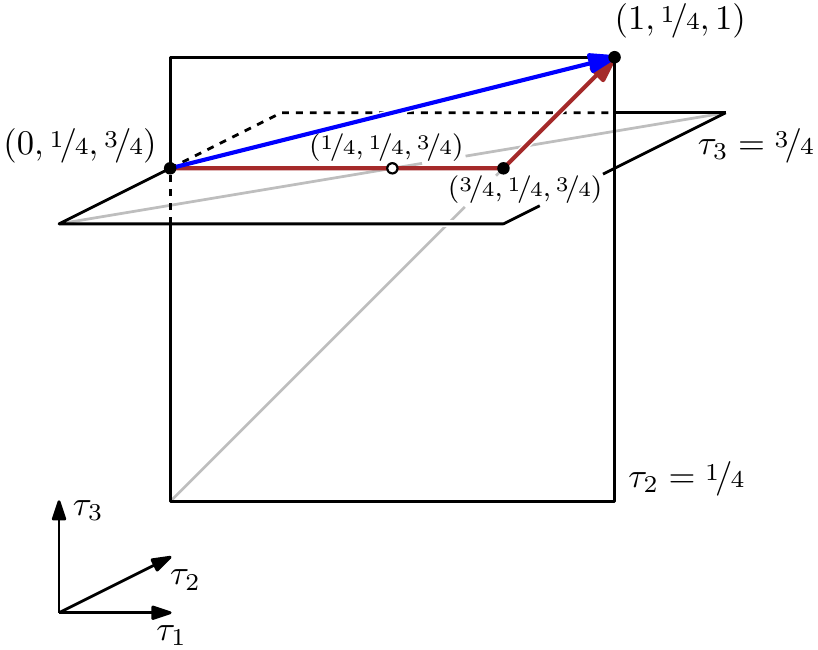}
    \caption{Three simplices, $\tsx_1, \tsx_2, \tsx_3$ have initial values
             $(0,\nicefrac{1}{4},\nicefrac{3}{4})$. Our goal is to increase the
             value of $\tsx_1$ to $1$, and we assume that its critical set
             includes simplex $\tsx_3$, but not $\tsx_2$.
             The final simplex values are $(1,\nicefrac{1}{4},1)$.
             The path taken by the gradient flow is shown in brown.
             The big step that our algorithm identifies, in blue.}
    \label{fig:big-step}
\end{figure}

\subsection{Increase Death}
\label{sec:increase-death}
Suppose we are trying to increase the value of $\tsx$, paired with $\ssx$, from $d$ to $d'$.
And suppose decomposition $D_p = R_pU_p$ is obtained using a lazy reduction.
Then it suffices to examine the row $\row{U_p}{\tsx}$ to identify the simplices
that must move together with $\tsx$. Specifically,
\begin{equation}
    \label{eq:increase-death}
    X_\ssx =
            \left\{ \tsx_i \;\middle\vert\;
            \begin{array}{@{}l@{}}
                d \leq f(\tsx_i) \leq d', \\
                U_p[\tsx,\tsx_i] \neq 0   
            \end{array}
            \right\}
\end{equation}
is the final critical set that we would accumulate under the gradient flow.
In other words, it suffices to move simplices in $X_\ssx$ --- and their cofaces
--- directly by setting $f(\tsx_i) = d'$.

\begin{theorem}
    \label{thm:increase-death}
    The critical set $X_\ssx$ defined in \cref{eq:increase-death} is the
    set of simplices accumulated by \cref{alg:transpositions}, when increasing
    the value of a negative simplex $\tsx$.
\end{theorem}
\begin{proof}
    Suppose there are $m_p$ $p$-simplices with $d \leq f(\tsx_i) \leq d'$.
    Denote the first $k$ of them with $Y_k$.
    We prove the claim by induction. Restrict the set $X_\ssx$ from
    \cref{eq:increase-death} to the set
    \begin{equation}
        \label{eq:increase-death-k}
        X_\ssx^k =
            X_\ssx \cap Y_k.
    \end{equation}
    We claim that
    this set is
    the same as its namesake in \cref{alg:transpositions}.

    The statement is trivially true for the base case: $X_\ssx^1 = \{ \tsx \}$.

    Consider the steps taken by \cref{alg:transpositions}. Suppose the claim is
    true after $k-1$ steps.
    By induction, all simplices $\tsx_i$ in $X_\ssx^{k-1}$ have
    $U_p[\tsx,\tsx_i] \neq 0$. Since the reduction is lazy,
    \cref{cor:lazy-U} implies
    $\ssx_i = \low \col{R_p}{\tsx_i} \leq \ssx$.
    At step $k$, we decide whether simplex $\tsx_k$ needs to be added to the critical set.

    \begin{figure}
        \centering
        \includegraphics{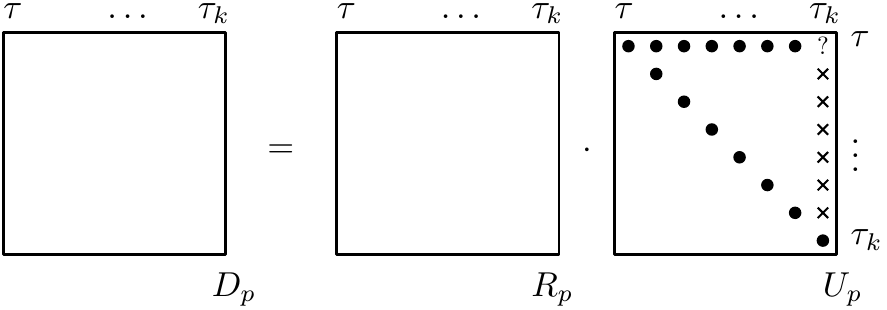}
        \caption{Subset of the matrices $D_p=R_pU_p$ involved in the proof of
                \cref{thm:increase-death}.}
        \label{fig:dru}
    \end{figure}

    Consider the subset of the $D_p=R_pU_p$ decomposition, restricted to the critical
    set and $\tsx_k$, i.e., simplices in the range $\tsx \ldots \tsx_k$; see~\cref{fig:dru}.
    We can zero out the column $\col{U_p}{\tsx_k}$ in this range using row operations in matrix
    $U_p$, adding multiples of row $\row{U_p}{\tsx_k}$ to the rows $\row{U_p}{\tsx_i}$
    above it. The corresponding operations in matrix $R_p$, which maintain the
    decomposition, subtract multiples of columns $\col{R_p}{\tsx_i}$ from column
    $\col{R_p}{\tsx_k}$. Denote the former by matrix $W$ and the latter by
    $W^{-1}$. We have $D_p = (R_p \cdot W^{-1}) \cdot (W \cdot U_p) = R_p' U_p'$.

    Once column $\col{U_p'}{\tsx_k}$ is zeroed out, we can
    transpose $\tsx_k$ with the critical set $X_\ssx^{k-1}$. The columns of the
    critical set may need to be reduced further, but the column
    $\col{R_p'}{\tsx_k}$ is already reduced, and therefore we can infer the pairing
    of $\tsx_k$ after the transposition.

    Denote by $\ssx_k = \low \col{R_p}{\tsx_k}$, the pair of $\tsx_k$ before the
    transposition.  If $\ssx_k > \ssx$, then it remains so
    after the transposition: by the inductive hypothesis
    $\ssx_i = \low \col{R_p}{\tsx_i} \leq \ssx$ for all $\tsx_i \in X_\ssx^{k-1}$,
    and therefore adding these columns to $\col{R_p}{\tsx_k}$ doesn't change its lowest
    non-zero. We note that because the reduction is lazy, in this case
    $U_p[\tsx,\tsx_k] = 0$ by \cref{lem:lazy-reduction}.

    If $\ssx_k < \ssx$, then we need to examine $U_p[\tsx, \tsx_k]$.
    If it is zero, then after the transposition $\ssx_k' = \low \col{R_p'}{\tsx_k}$
    remains less than $\ssx$, and therefore $\tsx_k$ does not become paired with $\ssx$.
    If $U_p[\tsx,\tsx_k] \neq 0$, then $\ssx = \low \col{R_p'}{\tsx_k}$ and $\tsx_k$
    enters the critical set.

    To summarize, $\tsx_k$ enters the critical set $X_\ssx^k$ if and only if
    $U_p[\tsx,\tsx_i] \neq 0$.
    In other words, $X_\ssx^k$ in \cref{eq:increase-death-k} and in
    \cref{alg:transpositions} are the same.

    It is crucial to our argument that if $\tsx_k$ does not enter the critical
    set, and therefore moves past it, that $U_p[\tsx,\tsx_k] = 0$. Because of this
    property, the row $\row{U_p}{\tsx}$ does not change via
    matrix updates in the induction, and therefore the entries that we encounter
    in the row at any step are the same.

\end{proof}

\subsection{Decrease Death}
\label{sec:decrease-death}

Suppose we are trying to decrease the value of simplex $\tsx$ from $d$ to $d'$.
And suppose decomposition $R_p = D_pV_p$ is obtained using a lazy reduction.
Then it suffices to examine the column $\col{V_p}{\tsx}$.
Specifically,
\begin{equation}
    \label{eq:decrease-death}
    X_\ssx =
            \left\{ \tsx_i \;\middle\vert\;
            \begin{array}{@{}l@{}}
                d' \leq f(\tsx_i) \leq d, \\
                V_p[\tsx_i,\tsx] \neq 0 
            \end{array}
            \right\}
\end{equation}
is the final critical set that we would accumulate under the gradient flow.
In other words, it suffices to move simplices in $X_\ssx$ --- and their faces
--- directly by setting $f(\tsx_i) = d'$.

\begin{theorem}
    \label{thm:decrease-death}
    The critical set $X_\ssx$ defined in \cref{eq:decrease-death} is the
    set of simplices accumulated by \cref{alg:transpositions}, when decreasing
    the value of a negative simplex $\tsx$.
\end{theorem}
\begin{proof}
    Suppose there are $m$ simplices with $d' \leq f(\tsx_i) \leq d$.
    Denote the last $k$ of them with $Y_k$.
    We prove the claim by induction. Restrict the set $X_\ssx$ from
    \cref{eq:decrease-death} to the set
    \begin{equation}
        \label{eq:decrease-death-k}
        X_\ssx^k = X_\ssx \cap Y_k.
    \end{equation}
    We claim that
    this set
    is the same as its namesake in \cref{alg:transpositions}.

    The statement is trivially true for the base case: $X_\ssx^1 = \{ \tsx \}$.

    Consider the steps taken by \cref{alg:transpositions}. Suppose the claim is
    true after $k-1$ steps.
    By induction, all simplices $\tsx_i$ in $X_\ssx^{k-1}$ have
    $V_p[\tsx_i,\tsx] \neq 0$. Since the reduction is lazy,
    \cref{cor:lazy-V} implies $\ssx_i = \low \col{R_p}{\tsx_i} \geq \ssx$.
    At step $k$, we decide whether simplex $\tsx_k$ needs to be added to the critical set.

    \begin{figure}
        \centering
        \includegraphics{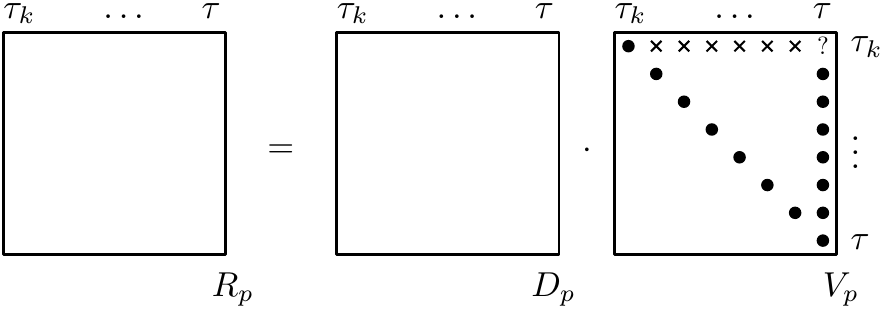}
        \caption{Subset of the matrices $R_p=D_pV_p$ involved in the proof of
                \cref{thm:decrease-death}.}
        \label{fig:rdv}
    \end{figure}

    Consider the subset of the $R_p=D_pV_p$ decomposition, restricted to the $\tsx_k$
    and the critical set, i.e., simplices in the range $\tsx_k \ldots \tsx$; see~\cref{fig:rdv}.
    Suppose we transpose $\tsx_k$ with all the simplices in the critical set
    $X_\ssx^{k-1}$, except for the last simplex $\tsx$. Denote the updated
    matrices $R_p'$ and $V_p'$. By \cref{cor:reordering},
    the pairing may change only among the transposed simplices. In particular,
    columns $\col{R_p'}{\tsx} = \col{R_p}{\tsx}$ and $\col{V_p'}{\tsx} = \col{V_p}{\tsx}$ do not change.

    If $V_p[\tsx_k,\tsx] = 0$, then we can transpose $\tsx$ and $\tsx_k$ without
    changing the pairing. In particular, $\tsx$ remains paired with $\ssx$. If
    $V_p[\tsx_k, \tsx] \neq 0$, then from the contrapositive of
    \cref{lem:lazy-reduction}, before the transposition
    $\ssx_k = \low \col{R_p}{\tsx_k} > \ssx$.
    From the inductive assumption
    (that together with \cref{lem:lazy-reduction} implies that for all
    $\tsx_i \in X_\ssx^{k-1} - \{ \tsx \}$, their pairs $\ssx_i > \ssx$)
    and from \cref{cor:reordering},
    after transposing $\tsx_k$ to just before $\tsx$, its pair
    $\ssx'_k = \low \col{R_p'}{\tsx_k} > \ssx$. To perform the final transposition, we need to
    zero out $V_p'[\tsx_k, \tsx]$, which adds a multiple of column
    $\col{R_p'}{\tsx_k}$ to $\col{R_p'}{\tsx}$. After the transposition, we undo the
    operation in the column of $\tsx_k$, which becomes
    \[
        \col{R_p}{\tsx_k} - (1/\alpha) \cdot (\col{R_p}{\tsx} - \alpha \cdot \col{R_p}{\tsx_k}) = - (1/\alpha) \cdot \col{R_p}{\tsx}
    \]
    where $\alpha = V_p[\tsx_k, \tsx]$. It follows that $\tsx_k$ becomes paired
    with $\ssx$ and therefore enters the critical set.

    To summarize, $\tsx_k$ enters the critical set $X_\ssx^k$ if and only if
    $V_p[\tsx_k,\tsx] \neq 0$.
    In other words, $X_\ssx^k$ in \cref{eq:decrease-death-k} and in
    \cref{alg:transpositions} are the same.

    It is crucial to our argument that if $\tsx_k$ does not enter the critical
    set, and therefore moves past it, that $V_p[\tsx_k,\tsx] = 0$. Because of this
    property, column $\col{V_p}{\tsx}$ does not change via
    matrix updates in the induction, and therefore the entries that we encounter
    in the column at any step are the same.
    This property, guaranteed by the use of the lazy reduction, is used in the
    proof via \cref{lem:lazy-reduction}.
\end{proof}

\begin{remark*}
    The proof of \cref{thm:decrease-death} carries through word-for-word if
    $\tsx$ is a positive unpaired simplex. This makes it possible to
    decrease the birth value of points at infinity by examining the respective
    column in matrix $V_p$. Notably, the argument breaks if simplex $\tsx$ is positive and
    paired. In this case the updates of the rows in matrix $R_p$ complicate the
    transpositions. It is not difficult to construct examples of the latter,
    where it is not enough to examine the columns of matrix $V_p$.
\end{remark*}

\subsection{Increase or Decrease Birth}
\label{sec:birth}

Thanks to duality, we are done.
Increasing and decreasing death in the previous subsection really means
moving $p$-simplex $\tsx$, with non-zero $\col{R_p}{\tsx}$, either to the left or to
the right in the filtration and matrices $D_p,R_p,V_p$, and $U_p$. In the dual
matrices $D_p^\bot,R_p^\bot,V_p^\bot$, and $U_p^\bot$, a simplex $\ssx$,
with non-zero $\col{R_p^\bot}{\ssx}$ is a birth simplex in a finite pair
$(\ssx,\tsx)$.  Moving it to the left in the anti-transposed matrices, whose
rows and columns are ordered in the reverse filtration order,
translates to \emph{increasing} its value in the filtration. Moving the simplex
to the right, to \emph{decreasing} its value.

As a result we get the following two theorems by substituting the dual matrices
into the proofs of \cref{thm:increase-death,thm:decrease-death}.

\begin{theorem}
    Critical set
    \[
        X_\tsx =
            \left\{ \ssx_i \;\middle\vert\;
            \begin{array}{@{}l@{}}
                d \leq f(\ssx_i) \leq d', \\
                V_p^\bot[\ssx_i,\ssx] \neq 0 
            \end{array}
            \right\}
    \]
    is the set of simplices accumulated by \cref{alg:transpositions}, when
    increasing the value of a positive $(p-1)$-simplex $\ssx$ paired with $\tsx$.
\end{theorem}

\begin{theorem}
    Critical set
    \[
        X_\tsx =
            \left\{ \ssx_i \;\middle\vert\;
            \begin{array}{@{}l@{}}
                d' \leq f(\ssx_i) \leq d, \\
                U_p^\bot[\ssx,\ssx_i] \neq 0 
            \end{array}
            \right\}
    \]
    is the set of simplices accumulated by \cref{alg:transpositions}, when
    decreasing the value of a positive $(p-1)$-simplex $\ssx$ paired with $\tsx$.
\end{theorem}

\begin{remark*}
    The remark at the end of the previous subsection about examining the column
    $\col{V_p}{\ssx}$ to decrease the value of an unpaired simplex $\ssx$
    translates to examining the column $\col{V_p^\bot}{\ssx}$ to increase its
    value.
\end{remark*}

\cref{tbl:operations} summarizes which matrices participate in each
case.

\begin{table}
    \centering
    \caption{Summary of operations and their respective rows and columns for
            $(p-1)$-dimensional $\ssx$ and $p$-dimensional $\tsx$.}
    \begin{tabular}{ccc}
        \toprule
        Operation               & Row/column  & Extra   \\
        \midrule
        Increase birth ($\ssx$) in $(\ssx, \tsx)$    & $\col{V_p^\bot}{\ssx}$  & cofaces      \\
        Decrease birth ($\ssx$) in $(\ssx, \tsx)$    & $\row{U_p^\bot}{\ssx}$  & faces        \\
        Increase death ($\tsx$) in $(\ssx, \tsx)$    & $\row{U_p^{\phantom{\bot}}}{\tsx}$       & cofaces      \\
        Decrease death ($\tsx$) in $(\ssx, \tsx)$    & $\col{V_p^{\phantom{\bot}}}{\tsx}$       & faces        \\
        \midrule
        Increase birth ($\ssx$) in $(\ssx, \infty)$  & $\col{V_{p}^\bot}{\ssx}$  & cofaces      \\
        Decrease birth ($\ssx$) in $(\ssx, \infty)$  & $\col{V_{p}^{\phantom{\bot}}}{\ssx}$       & faces        \\
        \bottomrule
    \end{tabular}
    \label{tbl:operations}
\end{table}

\subsection{Consistency of Critical Sets}
\label{sec:consistency}

\begin{figure}
    \centering
    \includegraphics{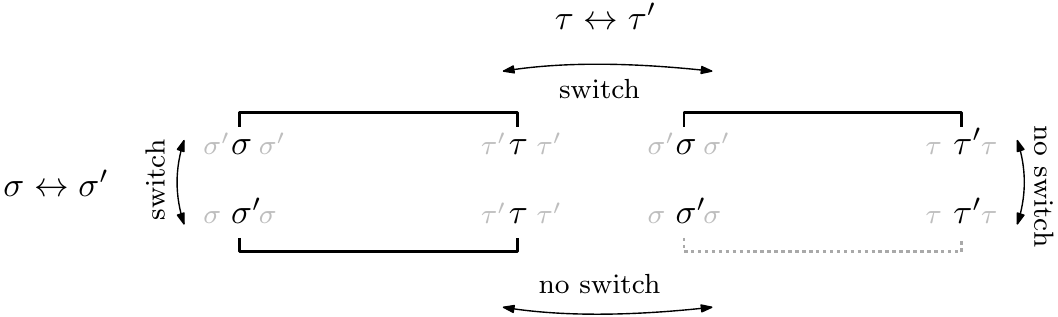}
    \caption{Top-left: simplices $\ssx$ and $\tsx$ are paired with each other
            and define the respective critical sets. $\ssx' \in X_\tsx$ and
            $\tsx' \in X_\ssx$ are contiguous with $\ssx$ and $\tsx$ (they come
            before if we are decreasing the respective value, and after, if we
            are increasing it).
            The assumption that simplices $\ssx'$ and $\tsx'$ are in critical sets implies
            that as we transpose them with $\ssx$ and $\tsx$, there is a switch in
            the pairing (top row and left column). If after both transpositions,
            $\ssx'$ and $\tsx'$ are not paired with each other (gray dotted line
            in the bottom-right), then no pairing switches during the
            second transposition. This determines the pairing between the four
            simplices: $\ssx$ is paired with $\tsx'$ and $\ssx'$ is paired with
            $\tsx$.}
    \label{fig:two-critical-sets}
\end{figure}

\cref{lem:critical-set-add,lem:critical-set} imply that individual critical sets
are well-defined: as we add simplices to a critical set during optimization, it
can never lose a simplex.  But what happens when we change
birth and death simultaneously? In this case, we have to settle for an
additional assumption, namely that point $p$ defining the singleton loss
has multiplicity one.

\begin{theorem}
    If $p = (f(\ssx), f(\tsx))$ has multiplicity one, then $X_\tsx$ and
    $X_\ssx$ don't change under permutation, i.e., for every $\tsx' \in X_\ssx$,
    if we swap it with $\tsx$, then its critical set $X_{\tsx'} = X_\tsx$, and
    for every $\ssx' \in X_\tsx$, if we swap it with $\ssx$, then its critical
    set $X_{\ssx'} = X_\ssx$.
\end{theorem}

\begin{proof}
Consider arbitrary simplices $\tsx'$ in $X_\ssx$ and $\ssx'$ in $X_\tsx$.
By \cref{def:critical-set}, if $\tsx$ and $\tsx'$ are swapped in the filtration,
then $\ssx$ is paired with $\tsx'$. Similarly, if $\ssx$ and $\ssx'$ are
swapped, then $\ssx'$ is paired with $\tsx$. Without loss of generality, we can
assume that both pairs of simplices --- $\ssx$ and $\ssx'$ as well as $\tsx$ and
$\tsx'$ --- are contiguous in the filtration. Then the two swaps above are
transpositions of contiguous simplices.
The assumptions about $\ssx' \in X_\tsx$ and $\tsx' \in X_\ssx$ imply that
either transposition (top and left in \cref{fig:two-critical-sets}) leads to a
switch in pairing. If after transpositions of both pairs, $\ssx'$
and $\tsx'$ are not paired with each other (lower-right part of
\cref{fig:two-critical-sets}), then no switch in pairing occurs during the
second of the two transpositions. This necessarily implies that $\ssx'$ is
paired with $\tsx$ and $\ssx$ is paired with $\tsx'$. In other words, there are
two persistence pairs between the critical sets, meaning point $p = (f(\ssx),
f(\tsx))$ in the diagram has multiplicity greater than one.
\end{proof}

\begin{remark*}
    The theorem applies to every critical set during the optimization. The point
    defining singleton loss may start out having multiplicity one, but gain
    higher multiplicity as the critical sets grow.
\end{remark*}

\subsection{Faces and Cofaces}
\label{sec:faces-cofaces}

After identifying the critical set $X_\ssx$, we need to move all the cofaces
(when increasing) or faces (when decreasing) of every simplex in the set, to
ensure that the new simplex values define a filtration. In
\cref{alg:transpositions}, the for-loop in
\crefrange{lin:face-coface-for-begin}{lin:face-coface-for-end}
performs the required update. For a general filtration, we simply execute the
same for-loop, which takes $\bigO(dm)$ time. Because we can identify the
critical set in $\bigO(m)$ time, finding the faces and cofaces dominates the
running time.
However, since $d$ is a small constant in all practical applications, linearity
in $m$ is most important.

In practice, an explicit update of faces and cofaces is unnecessary.
The function $f: K \to \Rsp$ that defines the filtration is derived from the input data.
The gradients on the simplices are backpropagated through $f$ to the input values, which are
updated by the optimization. When an updated function $f': K \to \Rsp$ is
derived from the updated data, it satisfies the face condition and defines a
filtration by construction.

For example, consider a lower-star filtration (which we use to compute
persistence of scalar fields in all our experiments in the next section):
given a function,
$\hat{f}: \Vert K \to \Rsp$, on the vertices of the simplicial complex, we
extend it to all the simplices, $f(\ssx) = \max \{ \hat{f}(v) \mid v \in \ssx \}$.
When we get a gradient $\partial \loss / \partial f(\ssx)$, which we backpropagate
to $\partial \loss / \partial \hat{f}(v)$,
where $v = \arg\max_{v' \in \ssx} \hat{f}(v')$. After taking a step, following
this gradient, we get a new function on the vertices $\hat{f}'$. Because the new
filtration is constructed as a lower-star filtration of this function, we are
guaranteed that all the faces precede $\ssx$ and all of its cofaces come after.
The same argument applies to the Vietoris--Rips filtration, {\v C}ech filtration, alpha filtration, etc.
In all such cases, the $\bigO(dm)$ term is eliminated from the running time,
leaving only $\bigO(m)$.

We note that it may still be worthwhile to compute and move the faces or cofaces
explicitly. The $\bigO(dm)$ overhead is minor, but more gradient information
gets propagated to the input data.

\subsection{Combined Loss}
\label{sec:combined-loss}

Given a general loss, $\loss = \sum_{(p,q) \in M} (p-q)^2$, defined by an
arbitrary matching $M$, typically recomputed after every step of the optimization,
we can compute the target values for each simplex
prescribed by the singleton losses defined by the individual terms of the sum.
For a simplex $\ssx$, with the initial value $f(\ssx) = a$,
we get a set of target values $\{ a_1, a_2, \ldots \}$, one for each singleton
loss. (If a singleton loss doesn't prescribe a value to a simplex, then the
corresponding value $a_i$ is missing from the target set, which can be empty as
a result.) Ultimately, we want to define a gradient on the individual simplices
that would allow us to take (small) optimization steps, but to do so we need to
decide on one target value for each simplex.

There are several ways to combine the target values into one. We choose to set
\[
    a' = a_i, \quad \mbox{where }i = \arg\max_j \{ \lvert a - a_j \rvert \}
\]
as the target value for $\ssx$, i.e., moving it as far as possible in the
filtration. (Two other strategies are considered in
\cref{sec:conflict-strategies}.)
This is a heuristic, without a strong justification, but with the
following reasoning behind it.
When all simplices of a given dimension are moving in the same
direction (e.g., all 1-simplices increase and all 2-simplices
decrease their values), most simplices get prescribed values $a_i$ that are lower bounds
on how far they need to move to solve the singleton loss. Put another way, all
but the first or the last simplex in the critical set can move farther than
their stated target. So taking the maximum is a way to satisfy all lower bounds
simultaneously. Another reason for the maximum is that for the simplification
loss $\loss_\ee$, defined in \cref{eq:simplification-loss} at the beginning of \cref{sec:singleton-loss}, when
applied to diagrams of dimension 0 or codimension 1, maximum gives the optimal
solution in one step. Intuitively, the reason is that when multiple values are
prescribed to the same simplex, it means that it belongs to multiple nested
topological features. (A formal proof of this claim requires a lot of new machinery,
which is why we omit it.  The claim is only a minor motivation for our heuristic
choice.)

\cref{alg:critical-method} summarizes our overall method.

\begin{algorithm}
    \caption{Critical set method.}
    \begin{algorithmic}[1]
        \State {\textbf{Input:} $\loss = \sum_{(p_i,q_i) \in M} (p_i - q_i)^2$}
        \For{\textit{each} $(p_i, q_i) \in M$}
            \State let $p_i = (b_i,d_i) = (f(\ssx_i), f(\tsx_i))$; $q_i = (b_i', d_i')$ \;
            \State $X_b = \left\{ \ssx_j \;\middle\vert\;
                        \begin{array}{l}
                            V^\bot[\ssx_j,\ssx_i] \neq 0 ~\textrm{and}~ b_i \leq f(\ssx_j) \leq b_i'; ~\textrm{or} \\
                            U^\bot[\ssx_i,\ssx_j] \neq 0 ~\textrm{and}~ b_i' \leq f(\ssx_j) \leq b_i \\
                        \end{array}
                    \right\}$ \;
            \State $X_d = \left\{ \tsx_j  \;\middle\vert\;
                        \begin{array}{l}
                            {\phantom{{}^\bot}}U[\tsx_i,\tsx_j] \neq 0 ~\textrm{and}~ d_i \leq f(\tsx_j) \leq d_i'; ~\textrm{or} \\
                            {\phantom{{}^\bot}}V[\tsx_j,\tsx_i] \neq 0 ~\textrm{and}~ d_i' \leq f(\tsx_j) \leq d_i \\
                        \end{array}
                    \right\}$ \;
            \State { // omitted: find faces/cofaces if necessary}
            \For{$\ssx_j \in X_b$}
                \State append $b_i'$ to $\target[\ssx_j]$
            \EndFor
            \For{$\tsx_j \in X_d$}
                \State append $d_i'$ to $\target[\tsx_j]$
            \EndFor
        \EndFor
        \For{\textit{each} $\ssx$}
            \If {$\target[\ssx]$ \textit{is empty}}
                \State $f'(\ssx) = f(\ssx)$ \;
            \Else
                \State {$j = \arg\max_j \left\{ \lvert f(\ssx) - \target[\ssx][j] \rvert \right\}$ \;}
                \State $f'(\ssx) = \target[\ssx][j]$ \;
            \EndIf
        \EndFor
    \Return{$\forall \ssx, \partial \loss / \partial f(\ssx) = 2(f(\ssx) - f'(\ssx))$}
    \end{algorithmic}
    \label{alg:critical-method}
\end{algorithm}

\begin{remark*}
    In practice, one typically uses an automatic differentiation library.
    Instead of computing the gradients explicitly, as in \cref{alg:critical-method},
    one computes the corresponding loss $\loss = \sum (f(\ssx) - f'(\ssx))^2$,
    where the summation is over all simplices $\ssx$ with non-empty $\target[\ssx]$.
    Value $f(\ssx)$ must be automatically differentiable and $f'(\ssx)$ is a constant.
    Backpropagation takes care of evaluating $\nabla \loss$
    with respect to $f(\ssx)$ and any variables on which $f(\ssx)$ depends.
\end{remark*}

\paragraph{Decreasing loss.}
In general, the heuristic of taking the maximum displacement as the target value
is not guaranteed to decrease the loss locally.
For such a guarantee, we need to assume that all the simplex values
are distinct and
take a sufficiently small step in any
direction whose individual components have the same sign as the negative of the gradient
of the loss. (The components on which the gradient of the loss is zero can
have any sign.) This follows from the loss being additive: each
individual term $(p_i - q_i)^2$ decreases if the coordinates of $p_i$ are
brought closer to $q_i$.

If one wanted to ensure that the loss decreases locally, it is easy to enforce
this condition explicitly by fixing the gradient values of the critical
simplices to be the same as the gradient given by the loss.
This results in an alternative heuristic for combining singleton losses. We
present one such example under the name \texttt{fca} in
\cref{sec:conflict-strategies}.

Another setting where the loss is guaranteed to decrease using
\cref{alg:critical-method} directly is if all points in the persistence diagram
are prescribed the same kind of movement, i.e., their birth and death values
either increase or decrease in tandem.

\section{Experiments}
\label{sec:experiments}

In this section, we compare optimization that computes gradients by identifying
critical sets of singleton losses, as explained in the previous section, to the
existing approach in the literature that defines the gradients on the pairs of
simplices that define the persistence pairing, as explained in the Introduction.
Below, in figure legends, we refer to our new method as ``Critical set'' and to
the previous method as ``Diagram.''

\paragraph{Vineyards.} In all our experiments
we get a series of diagrams ${D_t}$ indexed by the
optimization step. We visualize two of their projections to understand their
evolution.

\paragraph{Data.}
We use two scientific datasets from the ``Open Scientific
Visualization Datasets'' collection~\cite{osvd}.
\begin{itemize}
    \item Rotstrat~\cite{rotstrat}: temperature field of a numerical simulation of rotating stratified turbulence.
    \item Magnetic reconnection~\cite{magnetic_reconnection}:
        a single time step from a computational simulation of magnetic reconnection.
        This dataset has a visible geometric structure that looks like a curved
        tunnel in the middle.
\end{itemize}
We downsampled the data (still keeping it larger than most data sets used
for topological optimization in the literature).
All our experiments are done on datasets of size $32^3$.

We use upper- or lower-star filtrations, described in \cref{sec:faces-cofaces},
to compute persistence of the data. Because both use $\max$ or $\min$ to assign
values to the simplices, we apply the same maximum displacement construction as
in \cref{sec:combined-loss} to the vertices: if the same vertex is prescribed
different gradients by different simplices, we keep the one that results in the
largest displacement.

\subsection{Sublevel Set Simplification}
This experiment is motivated by the simplification of the decision boundary of a
neural network~\cite{chen2019}, formulated as a level set.
The authors phrase their loss in terms of well
groups~\cite{well-groups,well-groups-1d}.
For simplicity (to avoid having to introduce new constructions), we do not
simplify the level set, but rather a sublevel set.
Given a function $f: \Xsp \to \Rsp$, denote with $\Xsp_a =f^{-1}(\infty,a]$ its
sublevel set.
A topological feature exists in this sublevel set,
if its birth value is less than $a$ and its death value is greater than $a$.
Geometrically, we want to eliminate the points
of the persistence diagram that lie in the
quadrant defined by $b \leq a$ and $d \geq a$.
We match each such point $(b_i, d_i)$ to the closest point
on the boundary of the quadrant, i.e.,
either $(b_i, a)$ or $(a, d_i)$.

We ran this experiment for the magnetic reconnection dataset.
The threshold was chosen in such a way
that the quadrant contains a large portion of the points.

\cref{fig:dgm_wg_rs_klac_1_dgm,fig:dgm_wg_rs_klac_1_set_x} present the
vineyards of the two optimization procedures.
The color of the point encodes the step number, in both projections.
The blue lines show the quadrant that we want to make free of the diagram points.

\begin{figure*}
    \centering
    \includegraphics{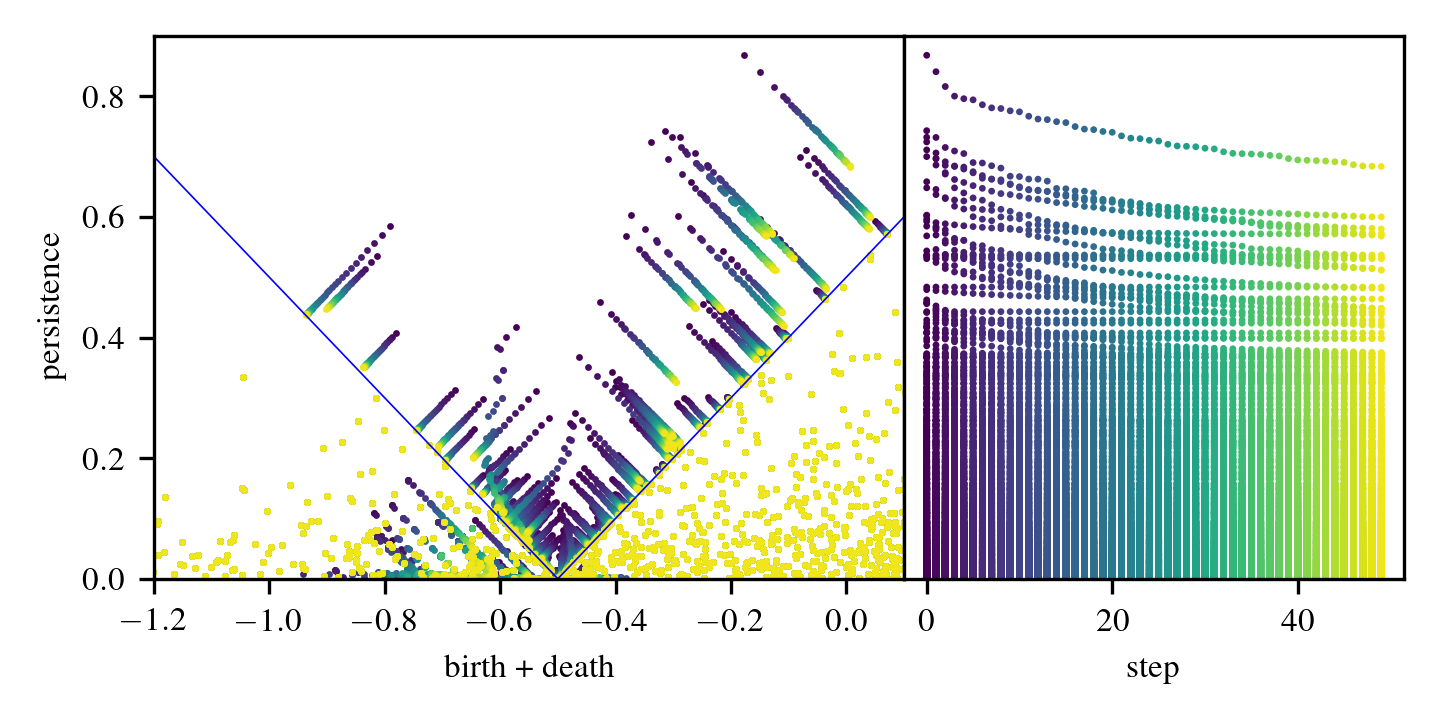}
    \vspace{-4ex}
    \caption{Vineyard of the optimization guided by the simplification of a
             sublevel set in a $0$-dimensional diagram of the magnetic
             reconnection dataset, using the \textbf{diagram method}.
             Learning rate is $0.2$, without momentum.
             The color encodes
             the time step. The left projection makes it clear that many of the
             points do not reach the quadrant boundary after 50 steps.}
    \label{fig:dgm_wg_rs_klac_1_dgm}
\end{figure*}

\begin{figure*}
    \centering
    \includegraphics{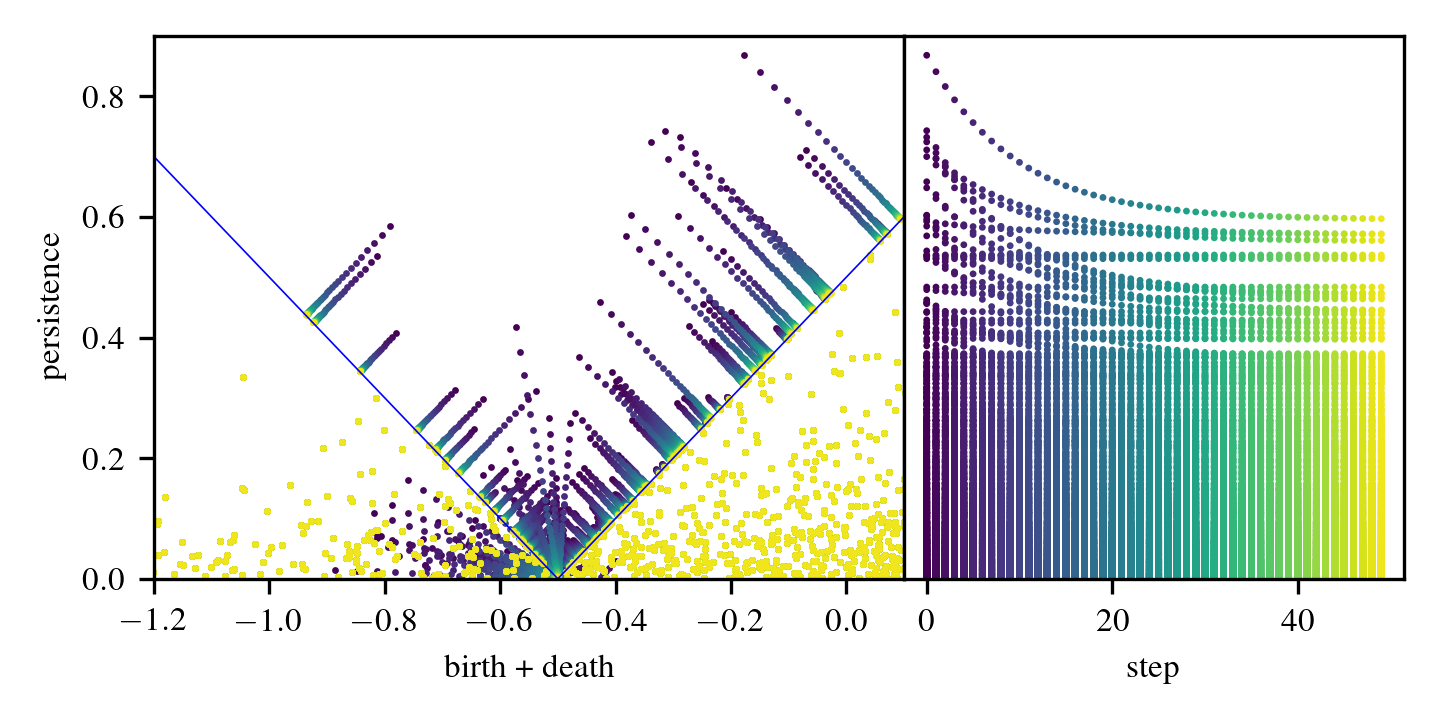}
    \vspace{-4ex}
    \caption{Vineyard of the optimization guided by the simplification of a
             sublevel set in a $0$-dimensional diagram of the magnetic
             reconnection dataset, using the \textbf{critical set method}.
             Learning rate is $0.2$, without momentum.
             The color encodes the time step. The left projection makes it clear that all the
             points rapidly reach the quadrant boundary.}
    \label{fig:dgm_wg_rs_klac_1_set_x}
\end{figure*}

Because topological features are intertwined in complicated ways,
it is impossible to move only the points in the quadrant.
The points outside of the quadrant are moving too, and
some of the points in the quadrant are not moving directly to their
prescribed target. This is expected in both cases.
What is notable is that using our critical set method,
the points move much more efficiently: after $50$ steps,
all points end up on the boundary of the quadrant, when using the critical set
method, but many do not reach the boundary, when using the diagram method.

\begin{figure*}
    \centering
    \includegraphics{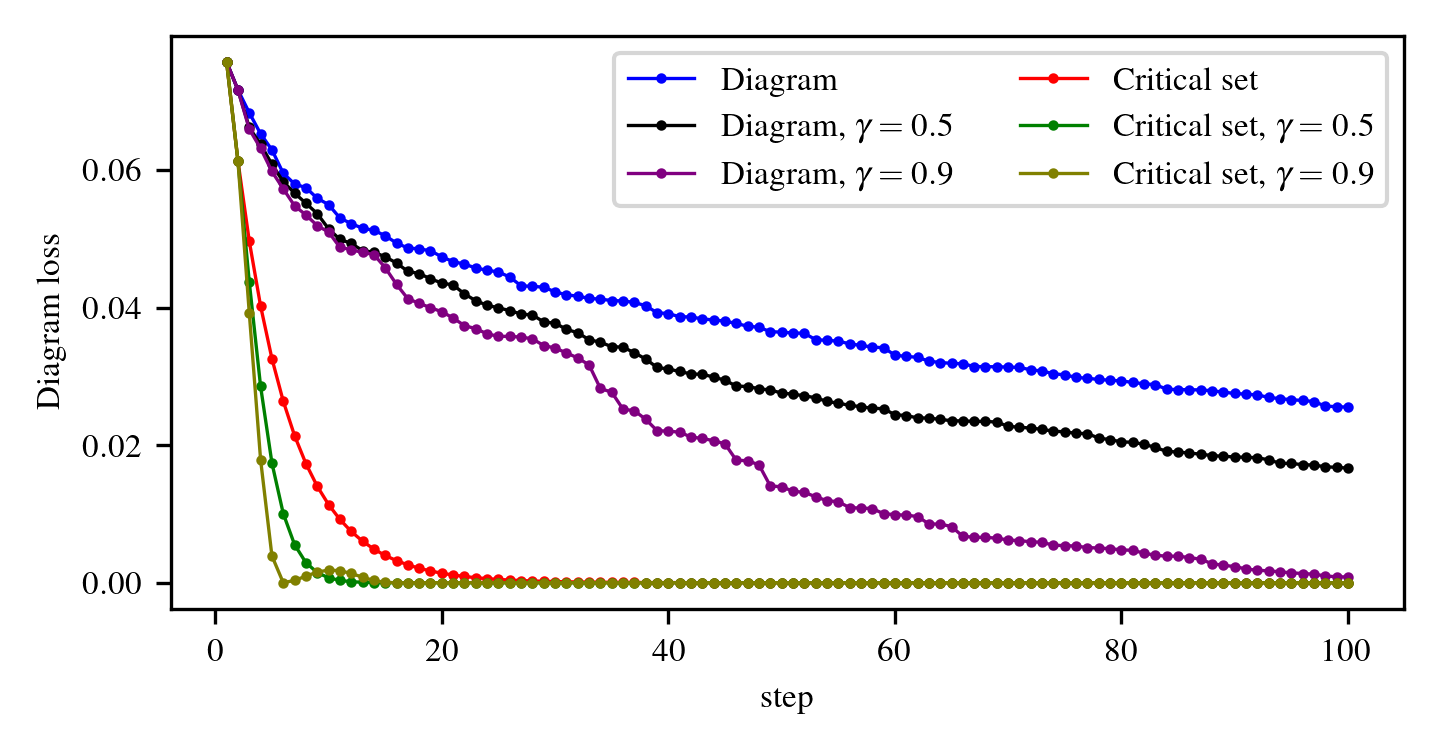}
    \vspace{-4ex}
    \caption{Comparison of the diagram losses during the optimization using the
             two methods. Diagram method greatly benefits from momentum.
             Critical set also benefits from momentum, but performs well even
             without it.}
    \label{fig:dgm_loss_klacansky_1}
\end{figure*}

To better compare the two optimization methods, we plot the value of the diagram
loss at each step of the optimization in \cref{fig:dgm_loss_klacansky_1}.
We used three optimization variants: standard gradient descent
and gradient descent with momentum, with damping parameter
$\gamma = 0.5, 0.9$. The smaller value of $\gamma$
makes the influence of the gradient from the previous
steps weaker.
Unsurprisingly, momentum makes a big difference for the diagram method: since it
needs to move large portions of the domain, but it has gradient information only
on the critical simplices, the ability to keep moving simplices for several
steps is crucial.
Our method also benefits from momentum, but less so, and it performs well with a lower value
of the damping parameter, $\gamma = 0.5$.
The diagram method works
best with the higher $\gamma = 0.9$, but even with this value
it is not nearly as fast the critical set method, which rapidly drops to $0$
with or without momentum.

\subsection{Persistence-sensitive Simplification}
Simplification loss was defined in \cref{eq:simplification-loss} at the beginning of \cref{sec:singleton-loss}:
it matches all the points with persistence below a prescribed threshold $\ee$ to
the diagonal.

We simplify the 1-dimensional diagrams, which is the case inaccessible to the
existing combinatorial methods.
The advantage of the critical set method is evident from 
the vineyards shown in \cref{fig:vin_dgms_dgm_loss_klacansky_rs_1,fig:vin_dgms_set_x_loss_klacansky_rs_1}.
The diagram method produces long trajectories of points moving towards
the diagonal. The critical set method moves the points much faster, which is
especially clear when comparing the right projections in the two figures.

\begin{figure*}
    \centering
    \includegraphics[]{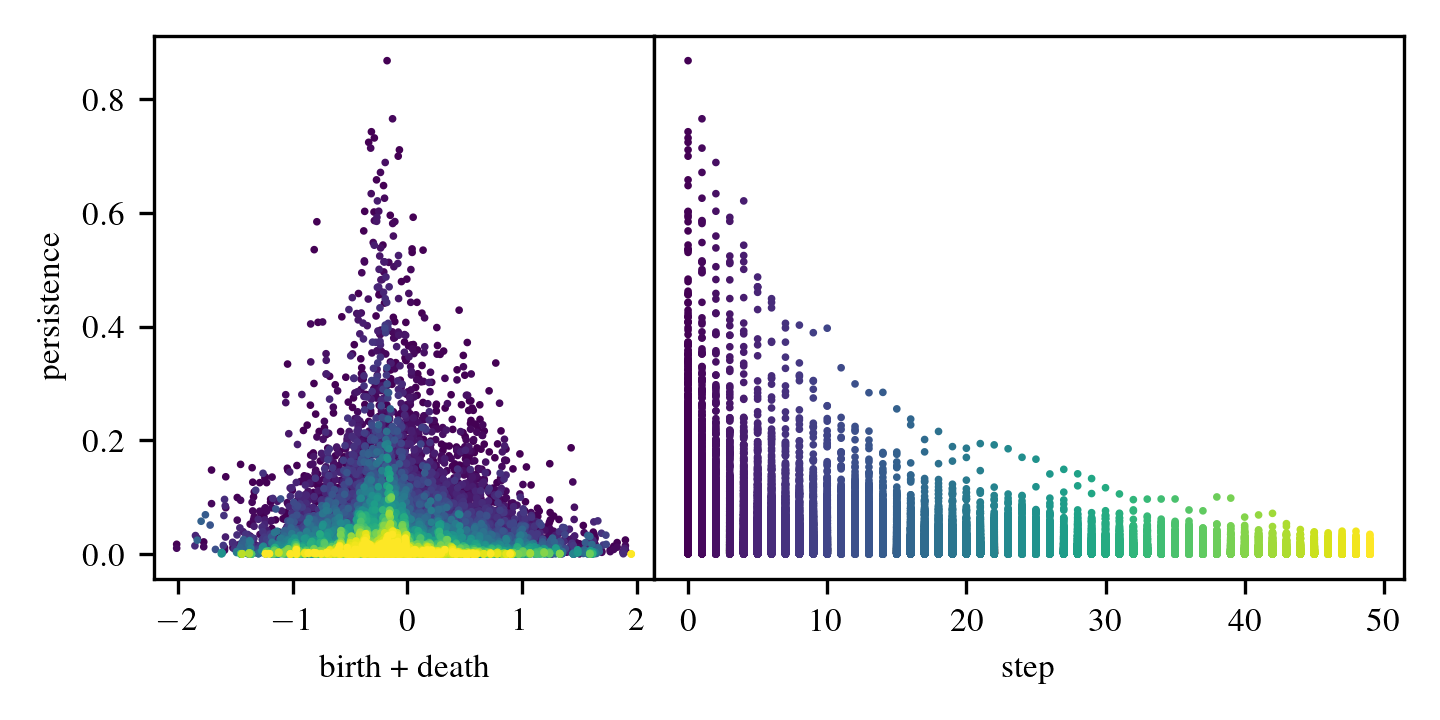}
    \vspace{-4ex}
    \caption{Vineyard of the optimization guided by the simplification loss in a
            1-dimensional diagram of the Rotstrat dataset, using the \textbf{diagram
            method}. The color encodes the time step. $\eps = \infty$, learning rate
            $0.2$, with momentum $\gamma = 0.9$.
            The plot makes it clear that the points don't reach their
            targets after 50 steps.}
    \label{fig:vin_dgms_dgm_loss_klacansky_rs_1}
\end{figure*}

\begin{figure*}
    \centering
    \includegraphics[]{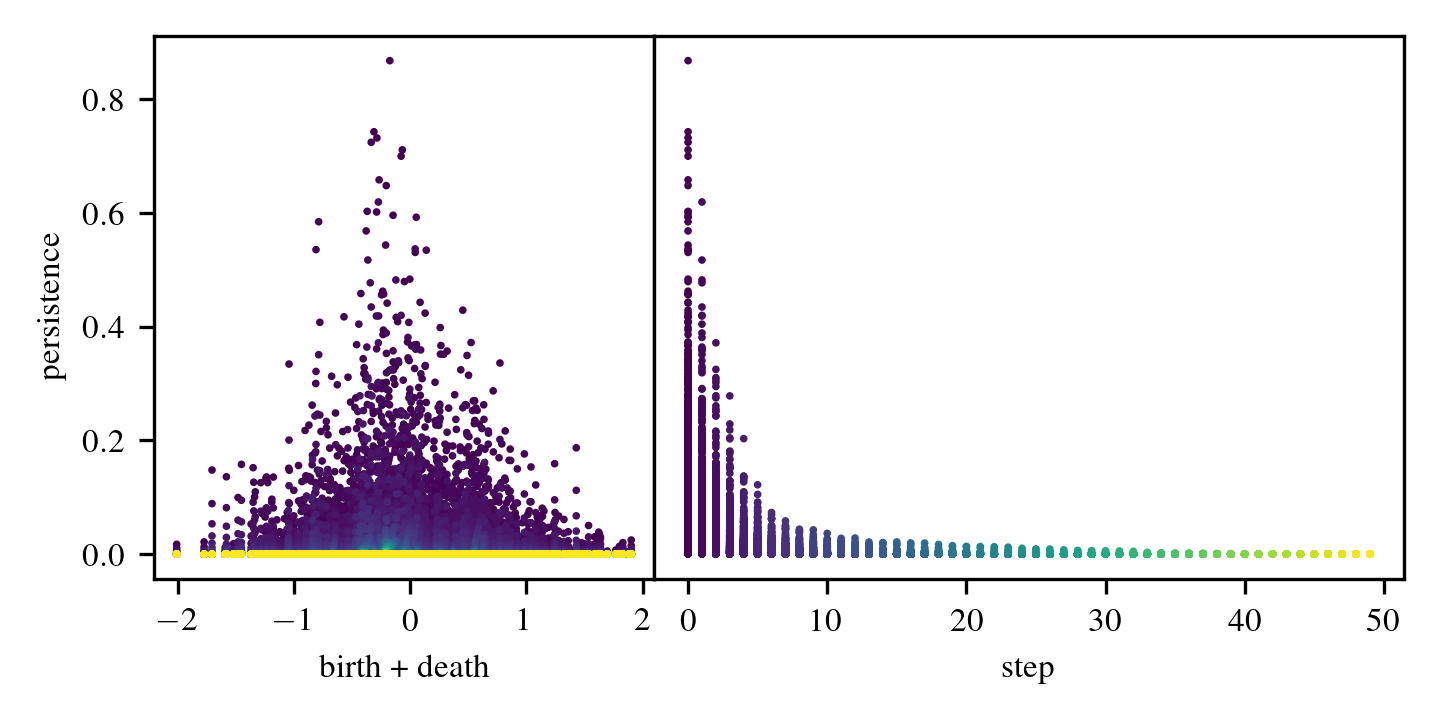}
    \vspace{-4ex}
    \caption{Vineyard of the optimization guided by the simplification loss in a
            1-dimensional diagram of the Rotstrat dataset, using the
            \textbf{critical set method}. The color encodes the time step.
            $\eps = \infty$, learning rate $0.2$, without momentum.
            The points approach the diagonal much faster than in \cref{fig:vin_dgms_dgm_loss_klacansky_rs_1}.}
    \label{fig:vin_dgms_set_x_loss_klacansky_rs_1}
\end{figure*}

The diagram loss plots are in \cref{fig:dgm_loss_klacansky_2}.
The $y$-axis is logarithmic, which emphasizes the advantage of the momentum
damping parameter of $\gamma=0.5$ for the critical set method.
Somewhat unexpectedly, a high value of momentum parameter ($\gamma = 0.9$)
almost completely wipes out the advantage of the critical set method.
For the diagram method, the momentum serves as a surrogate
for the critical set: it helps to further push the points,
which stopped being critical after one step.

\begin{figure*}
    \centering
    \includegraphics[]{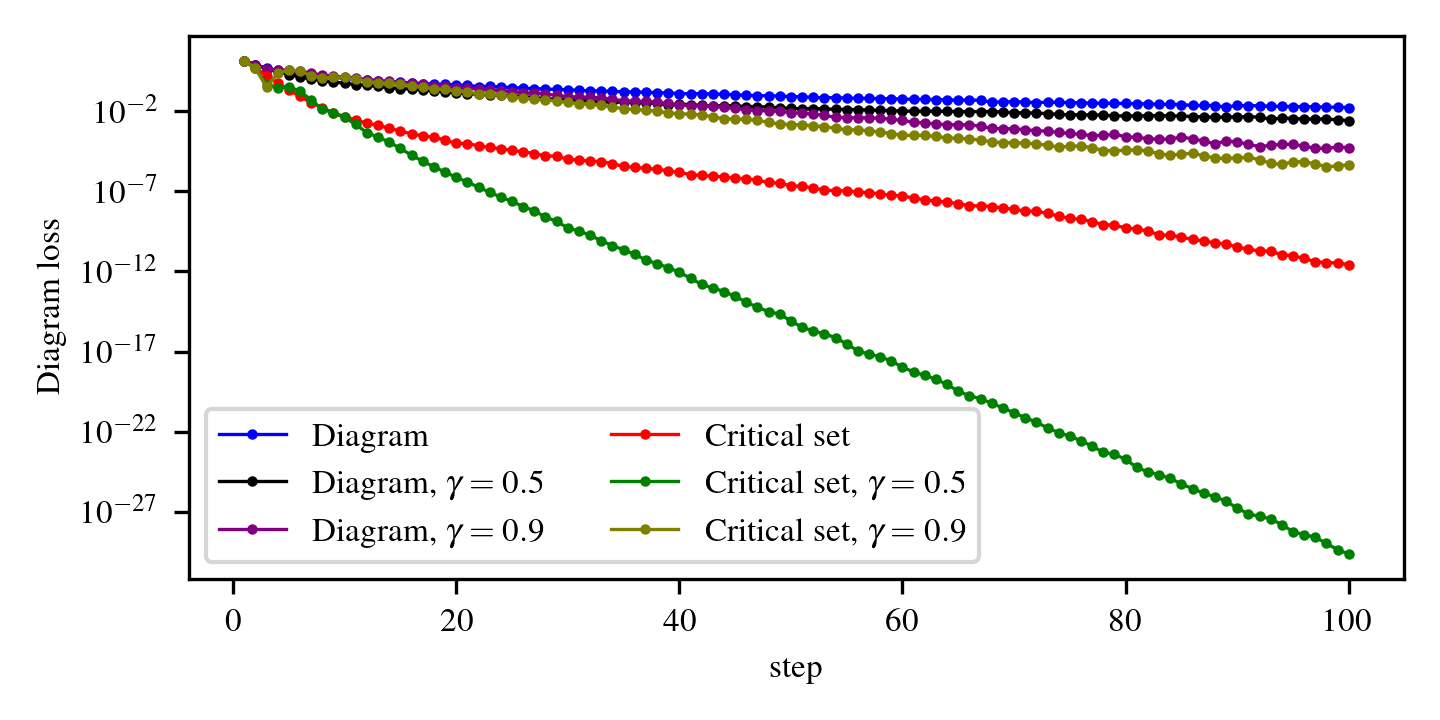}
    \vspace{-4ex}
    \caption{Comparison of the diagram losses during optimization of the
            simplification loss on Rotstrat dataset. Diagram methods benefits
            from momentum. So does critical set (for $\gamma=0.5$), but it
            performs much better than the diagram method, even without it.}
    \label{fig:dgm_loss_klacansky_2}
\end{figure*}

\subsection{Timing and Convergence Rate}
The major downside of our method is that it requires considerably more
computation per step.
We must compute not only the reduced boundary
matrix $R$, needed to read off the persistence diagram,
but also matrices $U$ and $V$. Moreover,
since we use both homology and cohomology, we have
to do this computation twice.

For example, for magnetic reconnection dataset it takes $3.4\times$ longer to
compute matrices $U$ and $V$ than matrix $R$ by itself. To compute all four
matrices $U,V,U^\bot$, and $V^\bot$ takes $4.2\times$ longer than just matrix
$R$.
Because it requires an order of magnitude fewer steps --- and the fraction gets
smaller as the data gets larger, see \cref{sec:scaling_experiments} ---
our method is still faster overall, but the result seems discouraging: much
of the savings suggested by the rapidly decreasing losses are lost because of
the more expensive per-step computation. We point out a possible solution in the
conclusion, but meanwhile note that some flexibility exists in the formulation
of the loss itself. For example, for the simplification loss, as we defined it,
we move every points $(b,d)$ to $((b+d)/2, (b+d)/2)$, which requires both to
increase birth and decrease death. As summarized in \cref{tbl:operations}, the
former requires computing matrix $V^\bot$; the latter, matrix $V$. But we could
also simplify the diagram by moving each point to the point $(b,b)$ on the
diagonal. This would require only decreasing the death values, and thus obviate
the need to compute cohomology.

\paragraph{Learning rate and momentum.}
To study the effect of the learning rate and momentum,
we simplify the Rotstrat dataset diagram in dimension $1$ for different values
of the hyper-parameters.
For $\eps = \infty$, the original value of the diagram loss is $7.2$.
For different values of the learning rate, we record the number of
steps needed to minimize it below $0.001$ using the two methods.
The results are in \cref{fig:epoch_lr_comparision}.

Without momentum, the critical set
method has a prominent advantage for all learning rates;
it requires $22-25\times$ fewer steps.
With momentum, the diagram methods performs better.
However, for large learning rates the performance of the best value of $\gamma = 0.9$
becomes worse: the corresponding purple line shoots up.
No choice of the hyper-parameters is a clear winner, but if we pick the two that
behave most reasonably --- $\gamma=0.5$ for the diagram method, and no momentum
for the critical set method --- we see about $11\times$ fewer steps for the
latter.

Taking into account the computational overhead, we conclude that the overall running time of our approach
is normally not worse than the diagram loss optimization, and for larger learning rates it is consistently
better.

\begin{figure}
    \centering
    \includegraphics[]{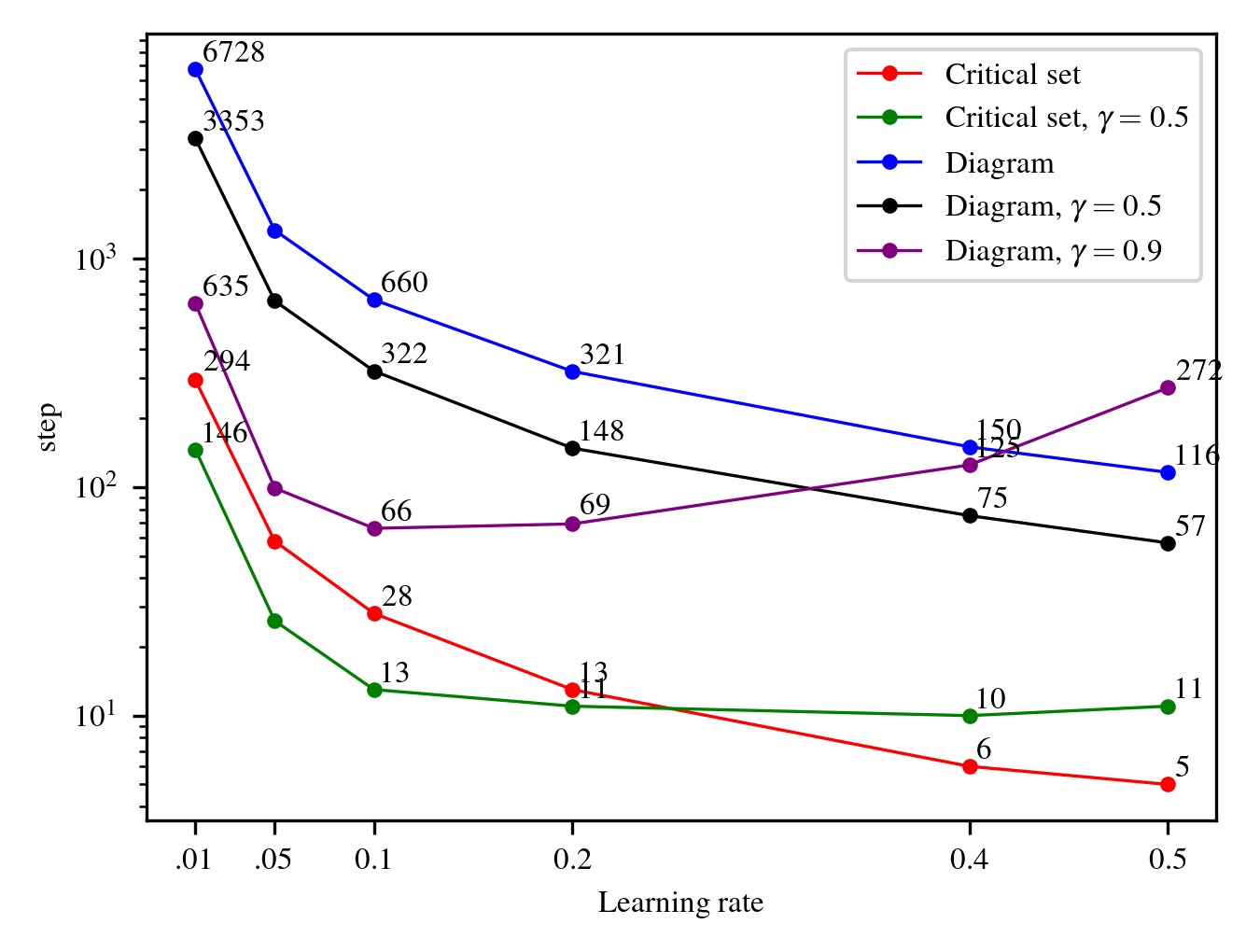}
    \vspace{-2ex}
    \caption{Comparison of the number of steps required to bring diagram loss
            from $7.2$ to $0.001$, using simplification loss on the 1-dimensional
            diagram of the Rotstrat dataset.  $\eps=\infty$.}
    \label{fig:epoch_lr_comparision}
\end{figure}

\section{Conclusion}
\label{sec:conclusion}

We have presented a method to accelerate optimization guided by a topological
loss, formulated as a matching. The method relies on examining the cycles,
chains, and related information calculated as a by-product of persistence
computation. We have shown empirically that our method reduces the number of
steps required to achieve a given loss by an order of magnitude.

\paragraph{Warm starts.}
The timing results seem discouraging: a
$10\times$ reduction in the number of optimization steps, combined with
a $4\times$ slow-down per step caused by the computation of matrices
$V,U$ and $V^\bot,U^\bot$ results in a very modest speed-up.
The fact that cohomology is not always needed provides little solace.
This may seem fatal to our approach, but the recent work of Luo and Nelson~\cite{LN21}
offers hope. Motivated by optimization, among other problems, they present a
simple algorithm to quickly compute persistence pairing, given a
reduction of a nearby filtration. They show that such ``warm starts''
significantly improve the computation speed, compared to recomputing the pairing
from scratch. Crucially for us, their algorithm relies on computing the $R=DV$
decomposition. In other words, following their method, there is no extra penalty
for computing matrix $V$, when iteratively updating persistence pairing.
Working out the technical details of such a combined approach is one the most
productive directions for future work.

\paragraph{Clearing optimization.}
Modern state-of-the-art implementations of persistence~\cite{ripser,phat},
use clearing optimization~\cite{clearing}, which identifies
zero columns of matrix $R$, corresponding to the births of finite pairs, without
reducing them explicitly. Such columns are not needed when we move points closer
to the diagonal --- increasing birth or decreasing death --- but the absence of
the corresponding operations in matrix $U$ presents a problem, when we want to
increase death or decrease birth. Although few of the losses proposed in the
literature need such operations, working out a complete method for combining our
construction with the clearing optimization is another worthwhile direction for
future work.

\paragraph{Combined losses.}
Given a matching, we combine multiple singleton losses by taking the maximum
displacement prescribed to individual simplices. This is a heuristic, without
a strong justification, other than what's stated in \cref{sec:combined-loss}.
There are other natural heuristics: for example, sending each simplex to the
average of its target values.
We discuss two of them in \cref{sec:conflict-strategies}.
Better
understanding the resulting dynamics and finding principled ways to combine multiple
singleton losses is the main theoretical question left open by our work.

In \cref{sec:hks_experiments}, we compare the performance of the diagram and
critical set methods on an optimization problem from \cite{Poulenard2018} that
back-propagates the loss past simplex values to a functional correspondence,
which serves as the parameter for optimization.

\paragraph{Convergence guarantees.}
Another major direction for future work is understanding convergence guarantees
of the critical set method. Carri\`ere et al.~\cite{Carriere2021} show that
persistence-based losses satisfy the assumptions required by the work
of Davis et al.~\cite{DDKL20}, and therefore their results on the convergence of
sub-gradient descent apply. Unfortunately, the gradient prescribed by the
critical set method does not lie in the sub-gradient of the loss. Its entire
point is to provide information on the non-critical simplices, on which the
sub-gradient of the original loss is necessarily zero.

\paragraph{Momentum and optimization.}
A striking result of our experiments is that momentum often hurts critical set
method, while it almost always helps the diagram method.
Our general intuition is that applied to the diagram method
momentum accumulates something like the critical set over the
iterations of the optimization. With the critical set method, the right
collection of simplices is identified by the algorithm itself and so bringing
information from prior iterations just obstructs progress.
Understanding the interaction of the critical set method with momentum,
and optimization more broadly, is another important research topic.

Experiments with other optimizers, namely RMSProp and Adam in
\cref{sec:different-optimizers}, reinforce that the interaction of momentum and
the critical set method is complicated and deserves future research.

\section*{Data Availability Statement}
The datasets analyzed are available in the
`Open Scientific Visualization Datasets' repository by P.\ Klacansky, \url{klacansky.com/open-scivis-datasets}.

\section*{Acknowledgments}
This work was initiated under Laboratory Directed Research and Development (LDRD)
funding from Berkeley Lab, provided by the Director, Office of Science, of the
U.S.\ Department of Energy under Contract No.\ DE-AC02-05CH11231.
It has since been supported by the U.S. Department of Energy, Office of
Science, Office of Advanced Scientific Computing Research, Scientific Discovery
through Advanced Computing (SciDAC) program and Mathematical Multifaceted Integrated
Capability Centers (MMICCs) program, under Contract No.\ DE-AC02-05CH11231 at Lawrence Berkeley National Laboratory.

\appendix
\appendixpage

\section{Different optimizers}
\label{sec:different-optimizers}

We also tried Adam and RMSProp optimizers.
The results for RMSProp are shown in \cref{fig:rmsprop_rotstrat}.
While there are some parameters which make direct diagram loss
optimization better than the critical set method, the
best results are still achieved when using the critical set.
The results for Adam in \cref{fig:adam_rotstrat} also demonstrate
the efficiency of the critical set method.

\begin{figure}
    \centering
    \includegraphics[]{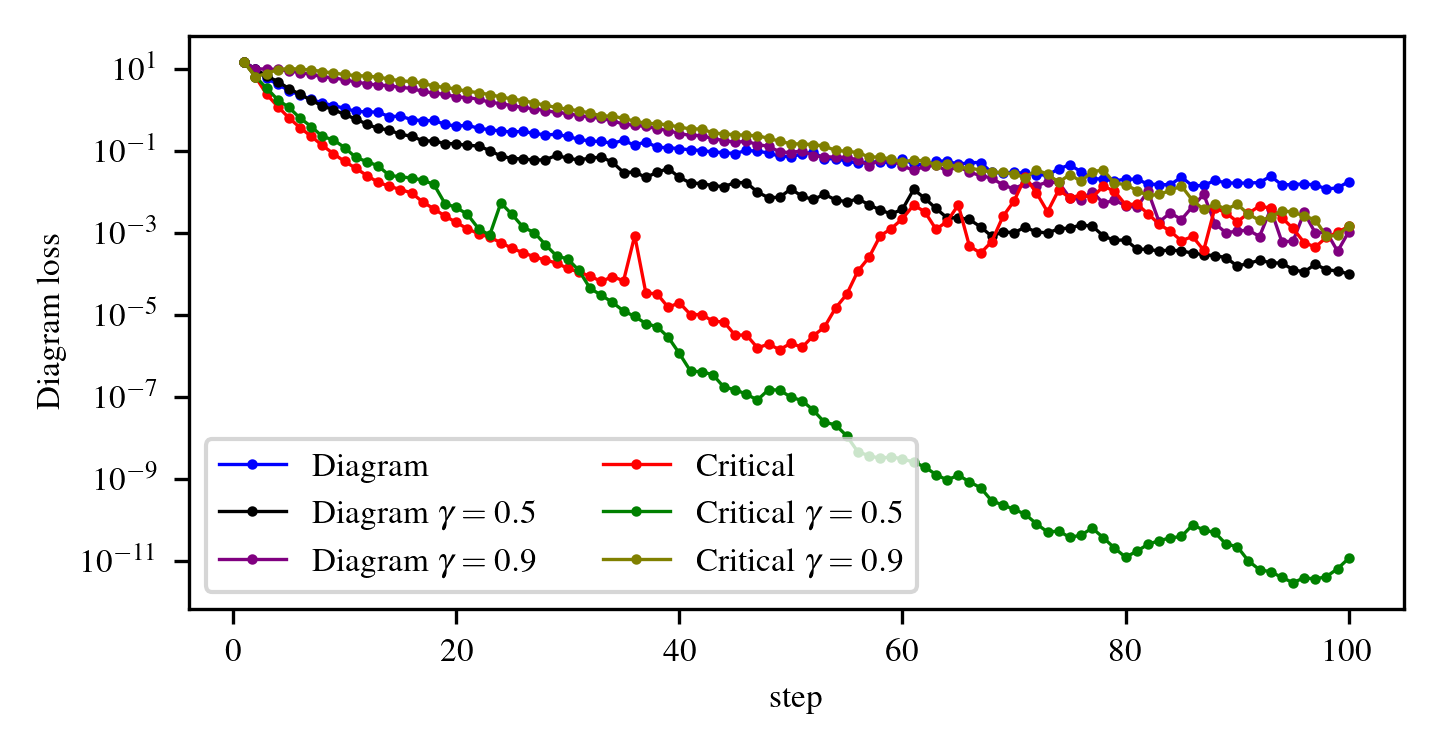}
    \vspace{-4ex}
    \caption{Simplification results for RMSProp on Rotstrat dataset.}
    \label{fig:rmsprop_rotstrat}
\end{figure}

\begin{figure}
    \centering
    \includegraphics[]{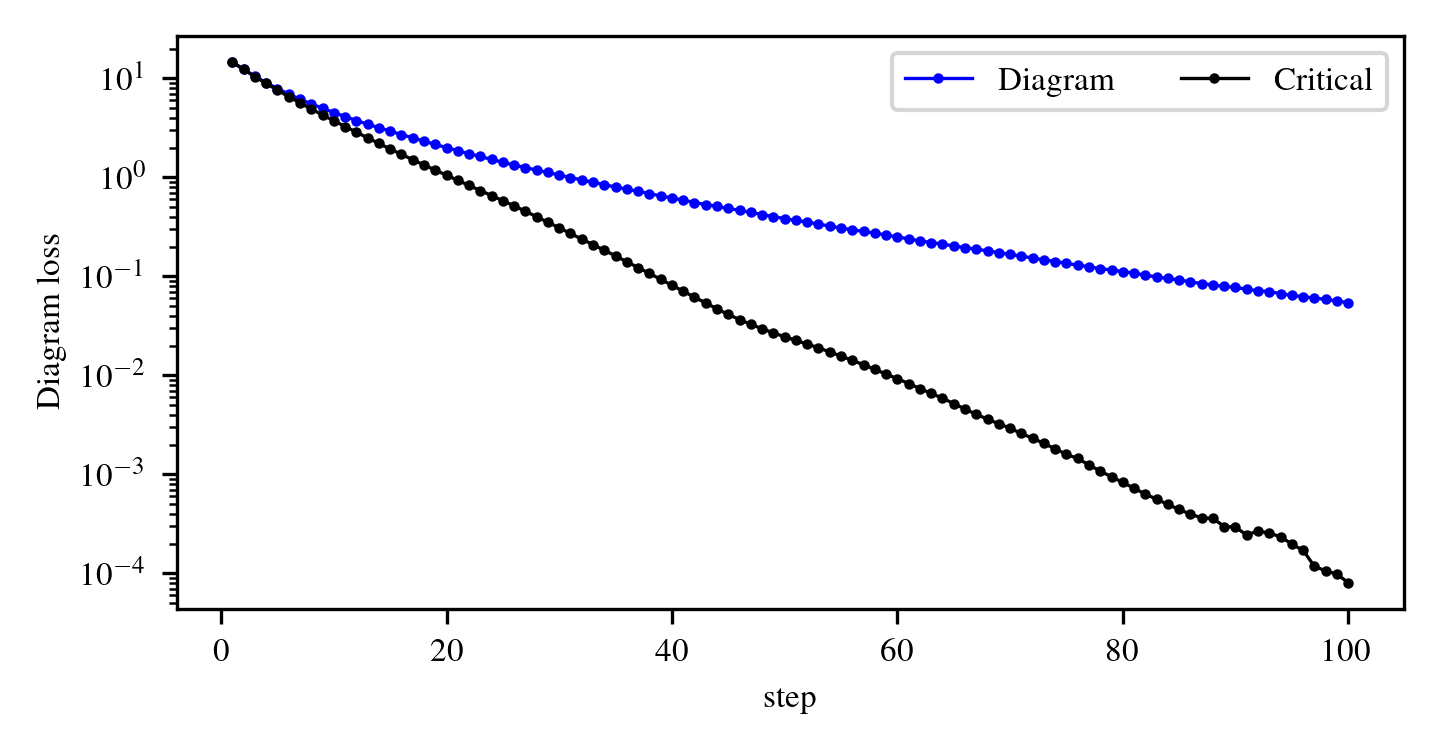}
    \vspace{-4ex}
    \caption{Simplification results for Adam on Rotstrat dataset. The parameters $\beta_1 = 0.9$ and $\beta_2 = 0.99$ are the default ones in PyTorch.}
    \label{fig:adam_rotstrat}
\end{figure}

\section{Other conflict strategies}
\label{sec:conflict-strategies}

We say that simplex $\ssx$ is in conflict, if
it belongs to multiple critical sets prescribed by the matching.
There are different ways to resolve such conflicts.
In \cref{sec:combined-loss}, we chose to take $v_k$ that maximizes $\lvert f(\ssx)-v_i \rvert$,
i.e., of all the values we take the farthest from the current one.
We abbreviate this choice as \texttt{max}.
Averaging is another natural choice:
if $\ssx$ appears in the critical sets $X_{\ssx_1}, \dots, X_{\ssx_k}$,
prescribing values $v_1, \dots, v_k$, then we assign $\frac{1}{k}\sum_{i=1}^k v_k$
as target value of $\ssx$. We abbreviate this choice as \texttt{avg}.

Another option is to take the average everywhere except the critical simplices.
Specifically, if $\ssx = \ssx_j$ is a simplex responsible for a point that appears in the
matching, then we assign $v_j$
as the target value. Otherwise we take the average. We abbreviate this method
as \texttt{fca} (Fix Critical simplices and take Average on others);
its pseudocode is in \cref{alg:critical-method-fca}.
This strategy imitates the gradient of the matching loss $\loss$.
Specifically,  \texttt{fca} guarantees
that for every critical simplex $\ssx_j$ whose point in the persistence diagram
appears in the matching, $\frac{\partial \loss}{\partial \ssx_j}$
and the $j$-th component of our gradient are the same.
If we assume general position, then the gradient of the diagram
loss is $0$ in all other components (infinitesimal perturbation
of other simplices does not change $\loss$).
In such general position, the loss $\loss$ is guaranteed to decrease in the
direction prescribed by the \texttt{fca} method, since the loss ignores values
of non-critical simplices.

\begin{algorithm}
    \caption{Fix Critical, Average on others.}
    \begin{algorithmic}[1]
        \State {\textbf{Input:} $\loss = \sum_{(p_i,q_i) \in M} (p_i - q_i)^2$}
        \For{\textit{each} $(p_i, q_i) \in M$}
            \State let $p_i = (b_i,d_i) = (f(\ssx_i), f(\tsx_i))$; $q_i = (b_i', d_i')$ \;
            \State $X_b = \left\{ \ssx_j \;\middle\vert\;
                        \begin{array}{l}
                            V^\bot[\ssx_j,\ssx_i] \neq 0 ~\textrm{and}~ b_i \leq f(\ssx_j) \leq b_i'; ~\textrm{or} \\
                            U^\bot[\ssx_i,\ssx_j] \neq 0 ~\textrm{and}~ b_i' \leq f(\ssx_j) \leq b_i \\
                        \end{array}
                    \right\}$ \;
            \State $X_d = \left\{ \tsx_j  \;\middle\vert\;
                        \begin{array}{l}
                            {\phantom{{}^\bot}}U[\tsx_i,\tsx_j] \neq 0 ~\textrm{and}~ d_i \leq f(\tsx_j) \leq d_i'; ~\textrm{or} \\
                            {\phantom{{}^\bot}}V[\tsx_j,\tsx_i] \neq 0 ~\textrm{and}~ d_i' \leq f(\tsx_j) \leq d_i \\
                        \end{array}
                    \right\}$ \;
            \State { // omitted: find faces/cofaces if necessary}
            \For{$\ssx_j \in X_b$}
                \State append $b_i'$ to $\target[\ssx_j]$
            \EndFor
            \For{$\tsx_j \in X_d$}
                \State append $d_i'$ to $\target[\tsx_j]$
            \EndFor
        \EndFor
        \For{\textit{each} $\ssx$}
            \If {$\target[\ssx]$ \textit{is empty}}
                \State $f'(\ssx) = f(\ssx)$ \;
            \ElsIf {$\ssx = \ssx_i$ for some $i$ (which is unique) }
                \State $f'(\ssx) = b_i'$ 
            \ElsIf {$\ssx = \tsx_i$ for some $i$ (which is unique) }
                \State $f'(\ssx) = d_i'$ 
            \Else
                \State $f'(\ssx) = \mbox{average of }\target[\ssx]$
            \EndIf
        \EndFor
    \Return{$\forall \ssx, \partial \loss / \partial f(\ssx) = 2(f(\ssx) - f'(\ssx))$}
    \end{algorithmic}
    \label{alg:critical-method-fca}
\end{algorithm}

\begin{figure}
    \centering
    \includegraphics[]{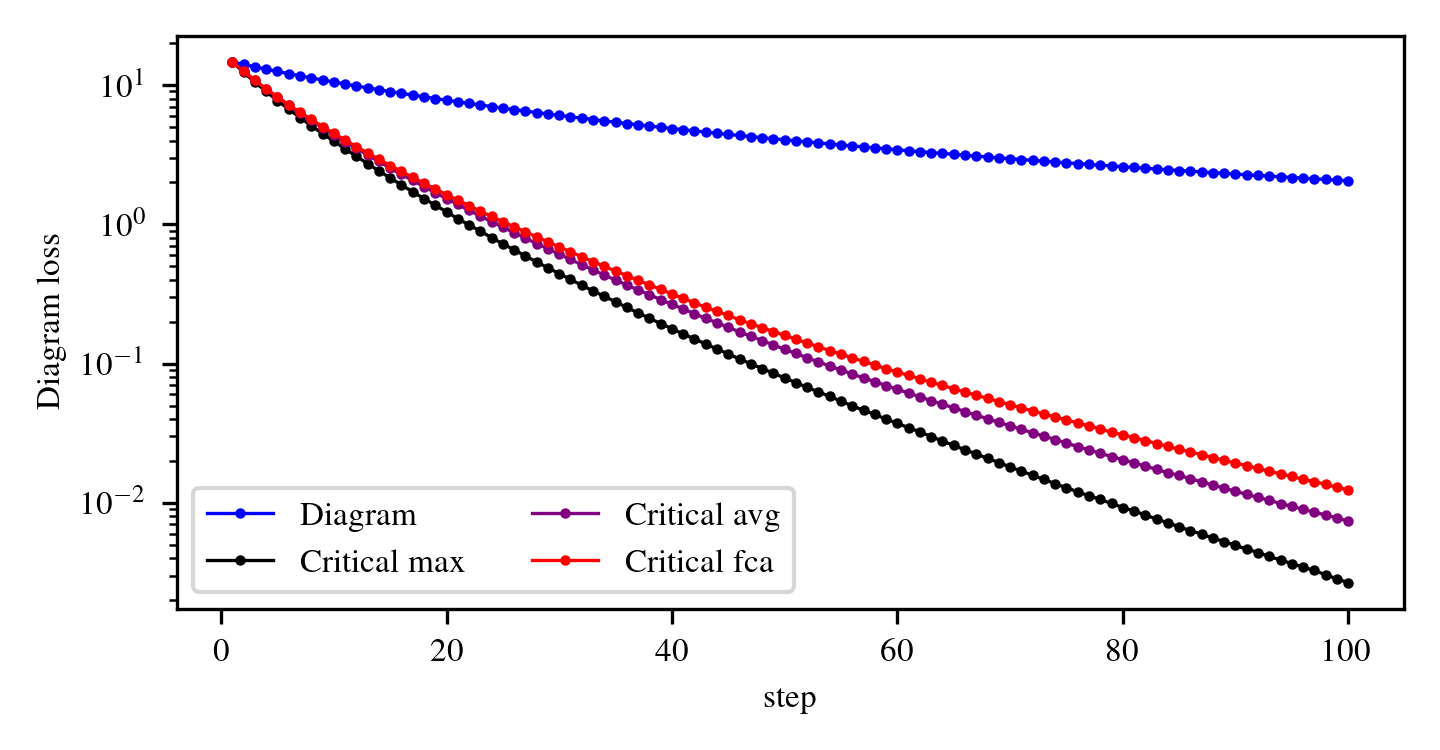}
    \vspace{-4ex}
    \caption{Comparison of different ways to combine singleton losses during
             optimization of the simplification loss on Rotstrat dataset, no momentum.}
    \label{fig:confl_strat_comp_no_mom}
\end{figure}

\begin{figure}
    \centering
    \includegraphics[]{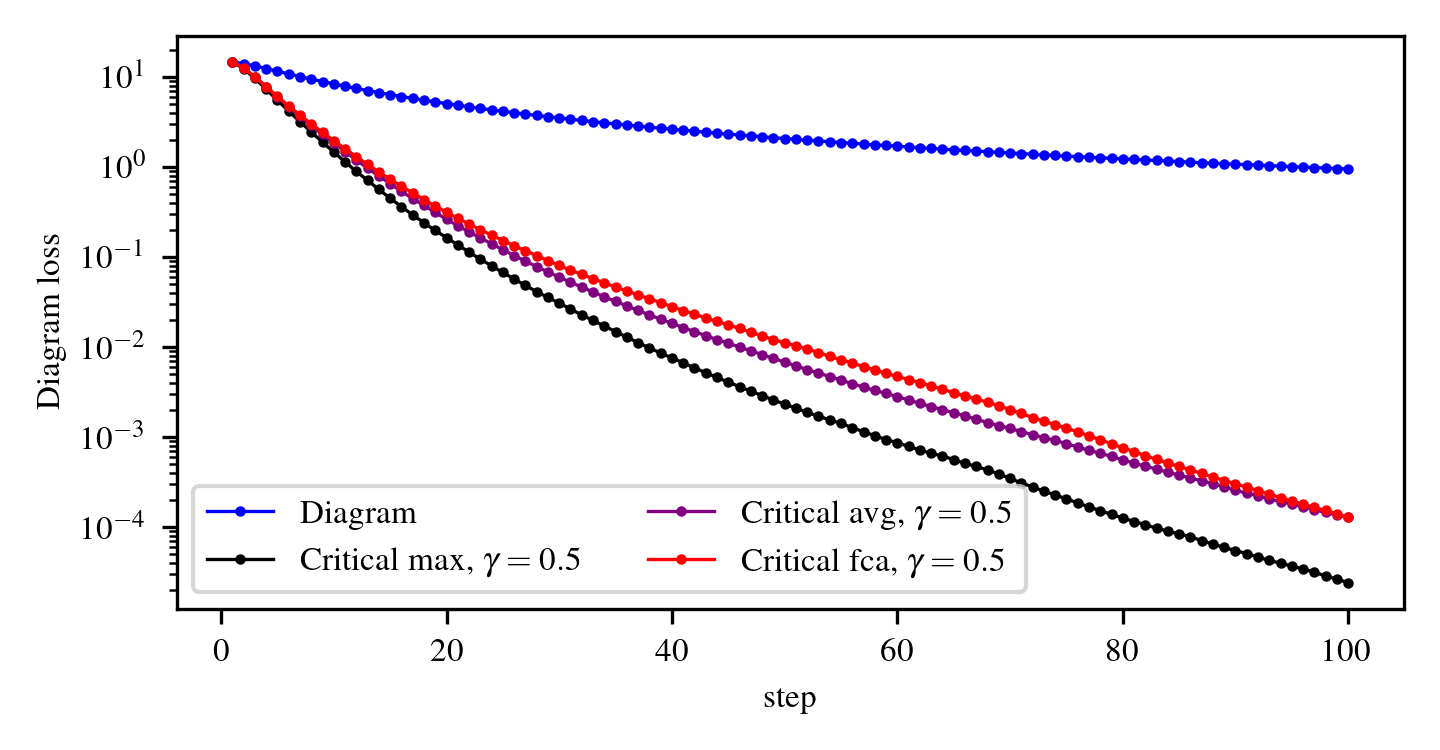}
    \vspace{-4ex}
    \caption{Comparison of different ways to combine singleton losses during
            optimization of the simplification loss on Rotstrat dataset, with momentum $0.5$.}
    \label{fig:confl_strat_comp_mom_0.5}
\end{figure}

\begin{figure}
    \centering
    \includegraphics[]{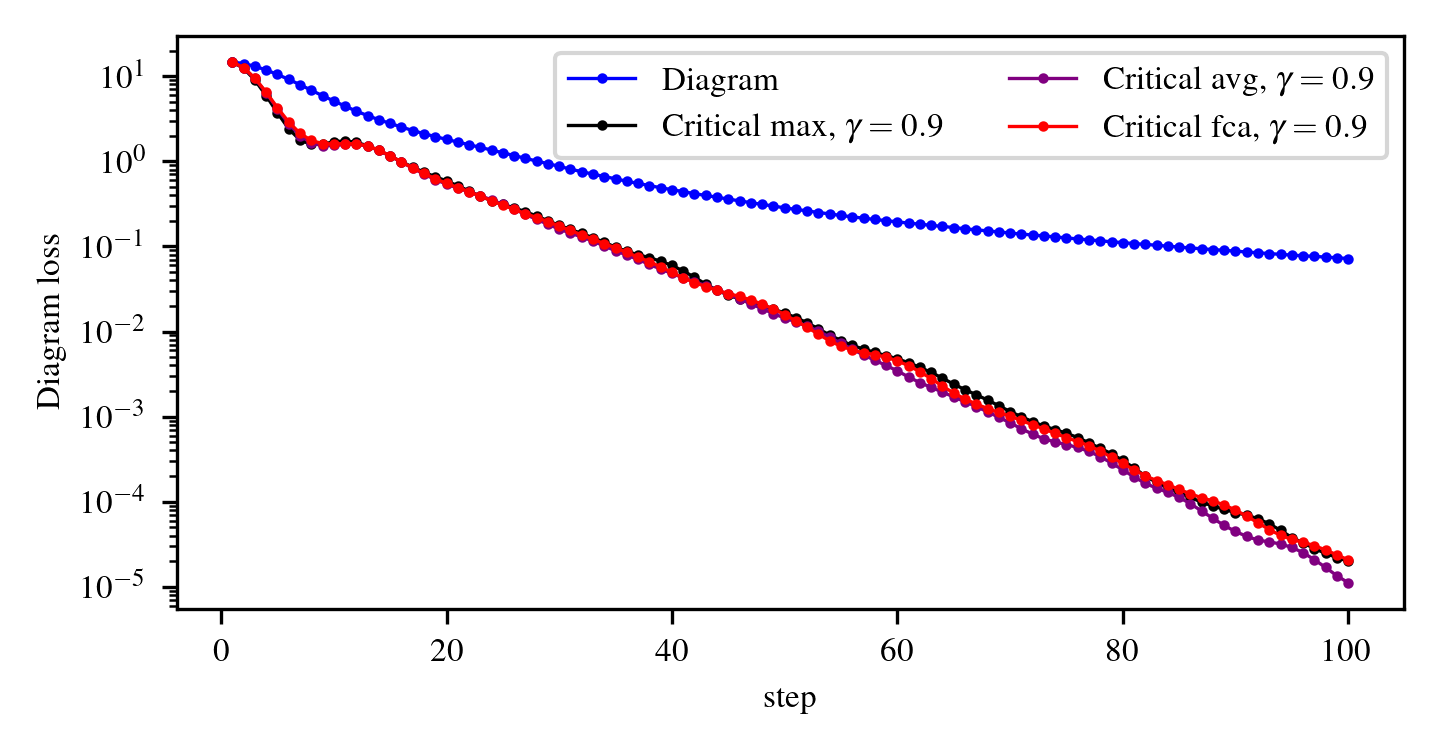}
    \vspace{-4ex}
    \caption{Comparison of different ways to combine singleton losses during
             optimization of the simplification loss on Rotstrat dataset, with momentum $0.9$.}
    \label{fig:confl_strat_comp_mom_0.9}
\end{figure}

We ran the Rotstrat example (simplification of $1$-dimensional diagram) with the same parameters
as in \cref{fig:dgm_loss_klacansky_2}, varying the conflict strategy.
\cref{fig:confl_strat_comp_no_mom,fig:confl_strat_comp_mom_0.5} show that there
is little difference between the three choices.
If $\gamma$ is small, then taking the maximum performs best.
For $\gamma = 0.9$, taking the average is slightly
better, see \cref{fig:confl_strat_comp_mom_0.9}.
In all cases, the critical set method clearly outperforms naive optimization of the diagram loss.

\section{Scaling experiments}
\label{sec:scaling_experiments}

Our method aims to move together
all the simplices whose critical values must be modified,
while the diagram method touches only the critical simplices.
Thus it is reasonable to expect
our method to perform better on larger inputs.
Suppose we want to simplify the diagram, and we have a version
of the same scalar field in different resolutions. For higher resolutions,
there will be more elements in the critical set of each point,
while the diagram loss identifies only one of these elements
at each step.

We took the Rotstrat example and downsampled it to 3 different sizes,
$32^3$, $64^3$ and $128^3$. Then we ran the well group simplification of $1$-dimensional diagram with the same parameters.
The plots of the losses are in
\cref{fig:scaling_no_momentum_losses,fig:scaling_with_momentum_losses}.
\cref{fig:scaling_with_momentum_losses} in particular shows that even when we
use momentum with the diagram loss, the critical set method
drives the diagram loss to zero significantly
faster for larger inputs. We also plot the ratio of the diagram loss values in \cref{fig:scaling_with_momentum_ratios}.
From this figure, we see that by step $50$, the diagram loss for the $32^3$ input was roughly $10^4$ times smaller when optimized
with the critical set method; for the $128^3$ input it was $10^9$ times smaller.

\begin{figure}
    \centering
    \includegraphics[]{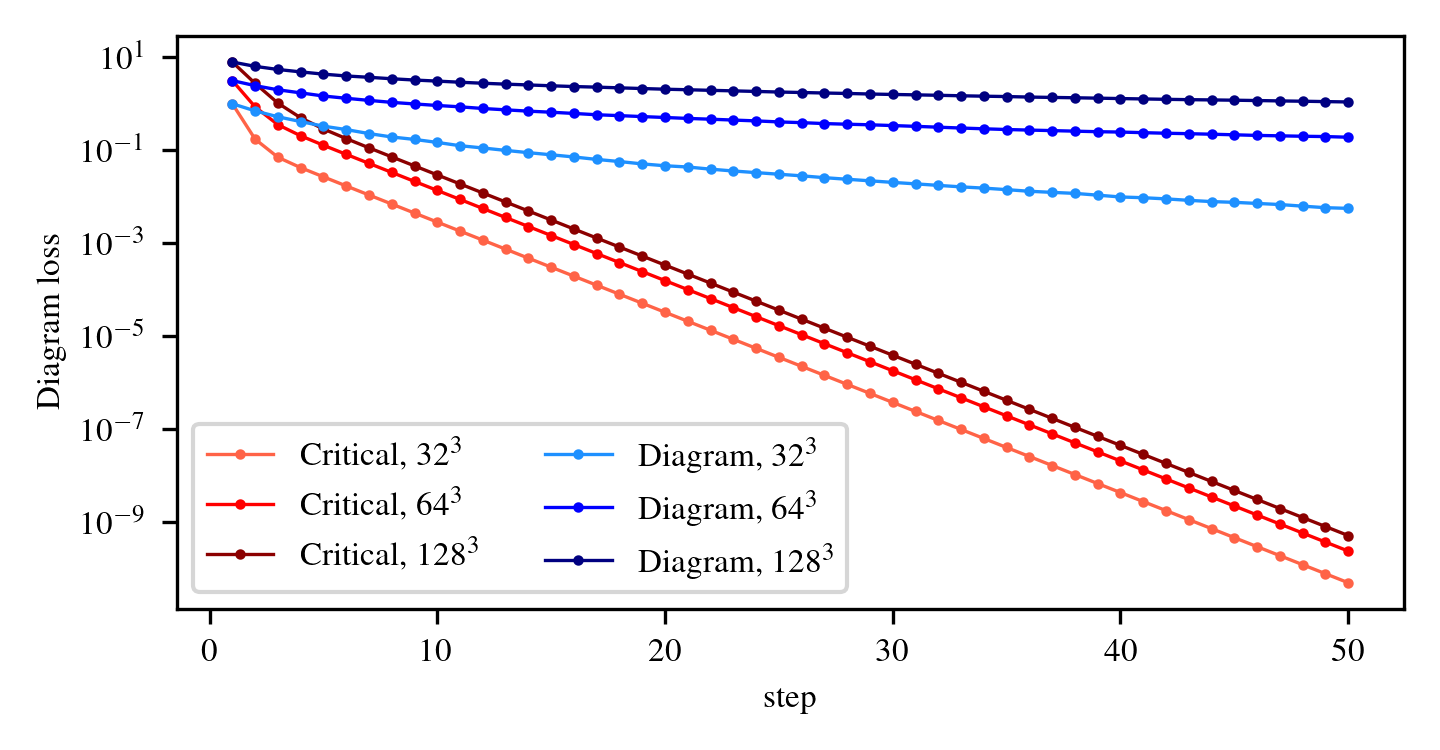}
    \vspace{-4ex}
    \caption{Comparison of the diagram losses during the optimization using the two methods on the inputs of different
    size.  The advantage of the critical set method becomes clearer for larger inputs. There is no momentum.
    Learning rate is $0.1$}
    \label{fig:scaling_no_momentum_losses}
\end{figure}

\begin{figure}
    \centering
    \includegraphics[]{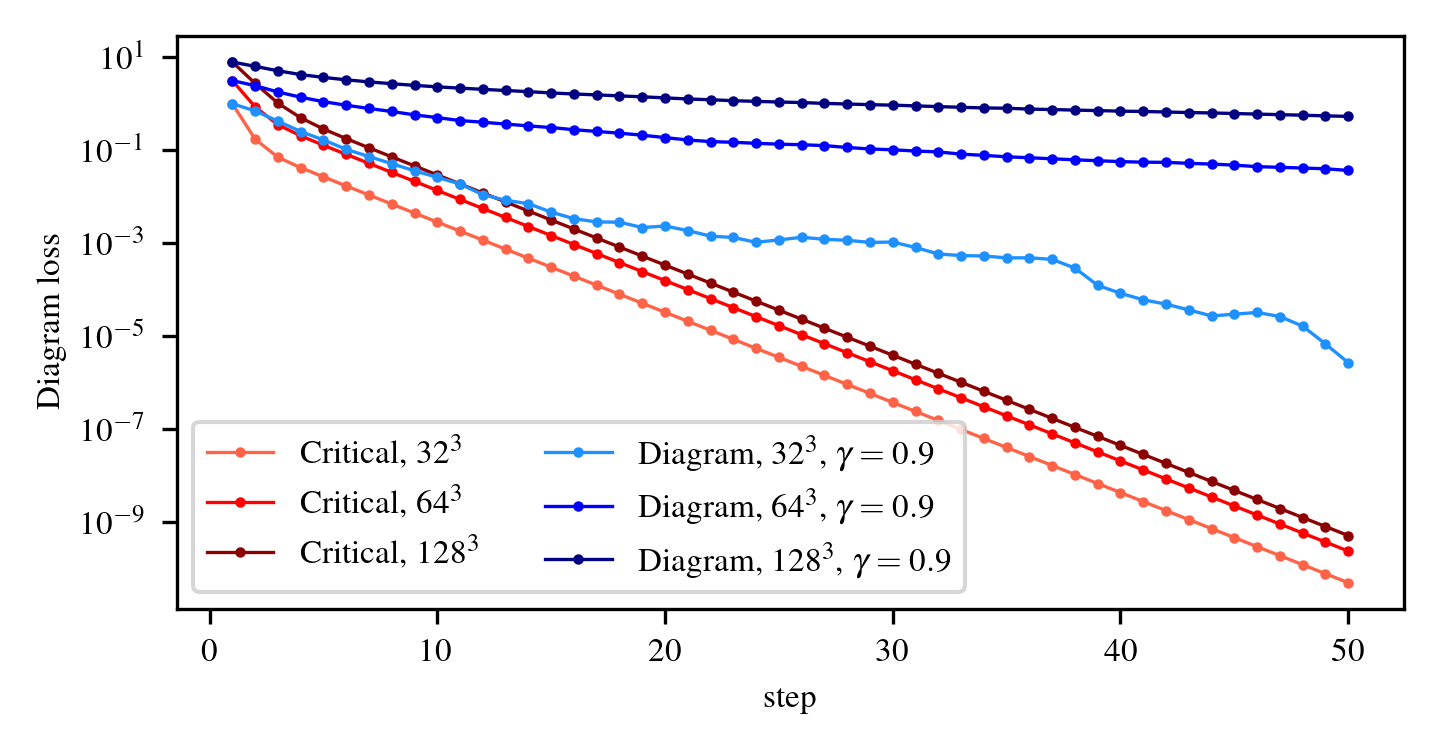}
    \vspace{-4ex}
    \caption{Comparison of the diagram losses during the optimization using the two methods on the inputs of different
    size.  The advantage of the critical set method becomes clearer for larger inputs. Optimization is with momentum for
    the diagram method, $\gamma=0.9$. Learning rate is $0.1$}
    \label{fig:scaling_with_momentum_losses}
\end{figure}

\begin{figure}
    \centering
    \includegraphics[]{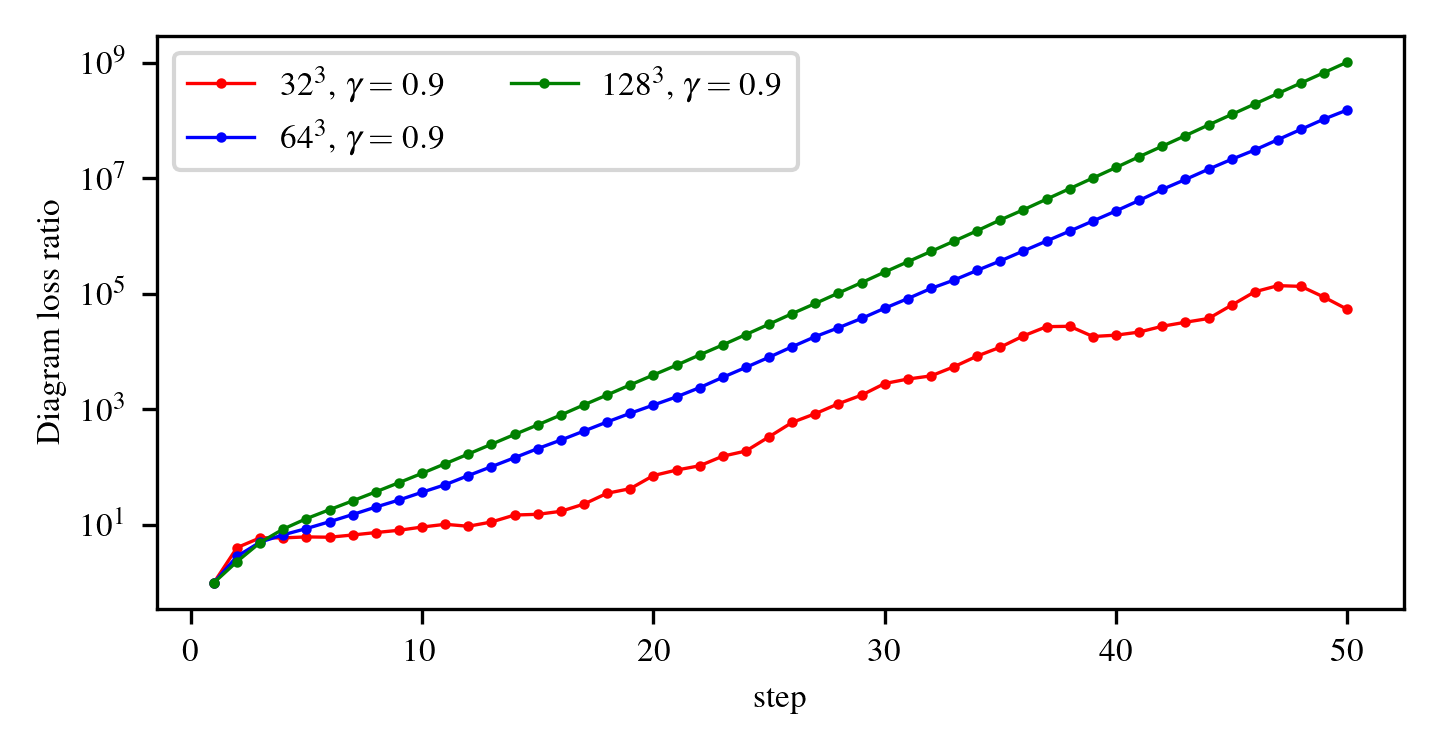}
    \vspace{-4ex}
    \caption{Ratio of the diagram losses during the optimization using the two methods on the inputs of different
    size: $y$-axis is the value of the diagram loss when optimized with the diagram method divided by the value of the diagram loss
    when optimized with the critical set method. Optimization is with momentum for
    the diagram method, $\gamma=0.9$. Learning rate is $0.1$.}
    \label{fig:scaling_with_momentum_ratios}
\end{figure}

\section{Experiments with Heat Kernel Signature}
\label{sec:hks_experiments}

We replicate some of the experiments from \cite{Poulenard2018}, both directly
optimizing the values on a mesh and back-propagating to optimize a functional
correspondence between two meshes.

\begin{figure}[ht]
    \centering
    \subfloat[Original HKS function.]{
        \includegraphics[width=0.3\textwidth]{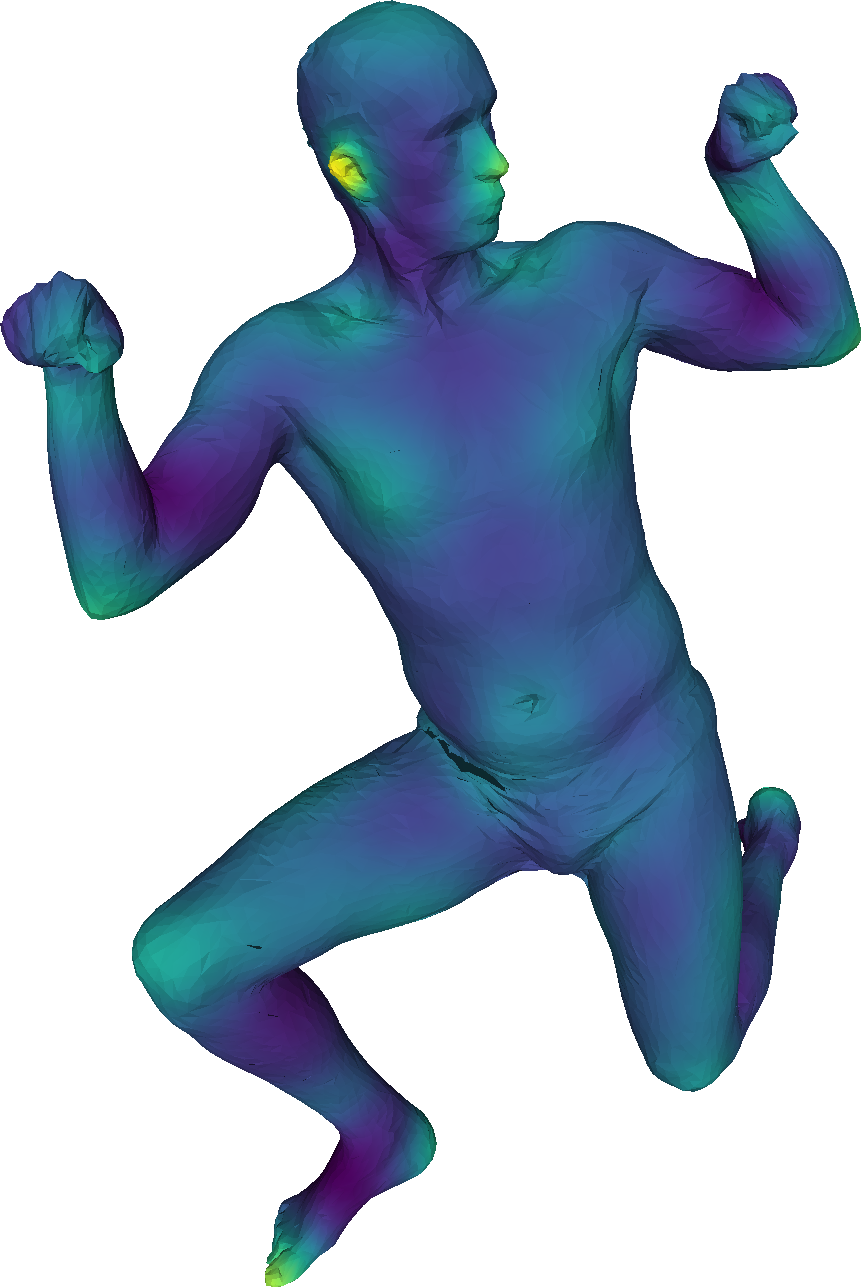}
    }
    ~
    \subfloat[Simplified (critical set).]{
        \includegraphics[width=0.3\textwidth]{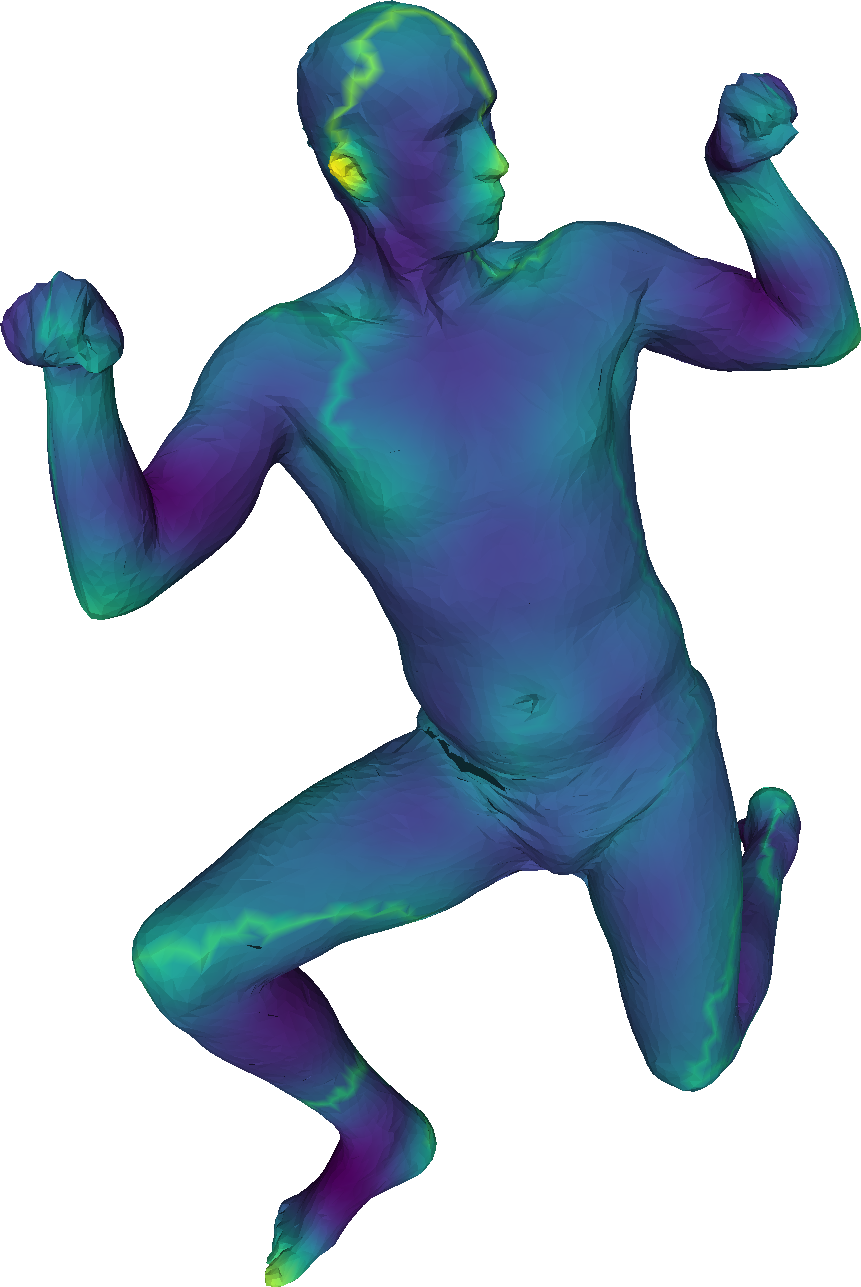}
    }
    ~
    \subfloat[Simplified (diagram).]{
        \includegraphics[width=0.3\textwidth]{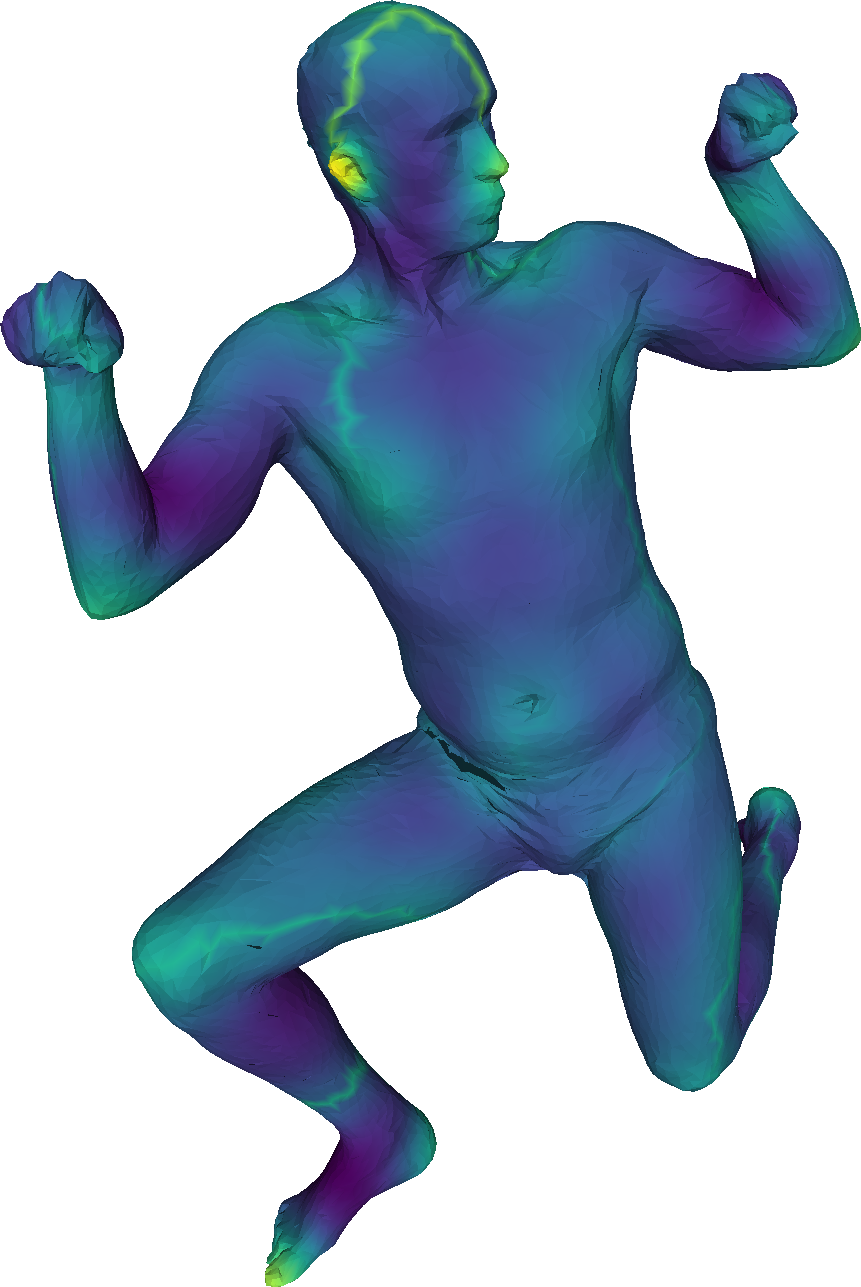}
    }
    \caption{Visualization of the Heat Kernel Signature and its simplification.}
    \label{fig:hks_result_human}
\end{figure}

\begin{figure}[ht]
    \centering
    \includegraphics[]{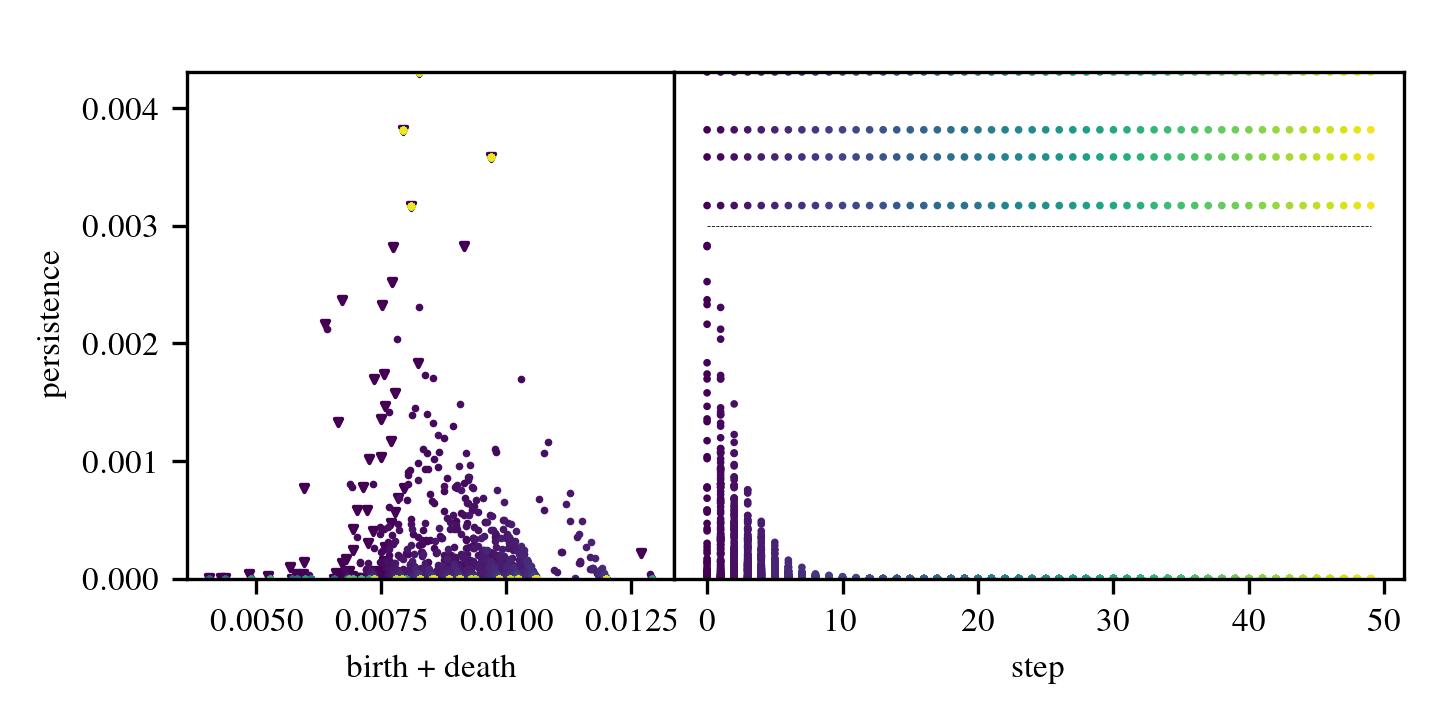}
    \vspace{-4ex}
    \caption{Vineyard of the optimization guided by the simplification of a
             superlevel set in a $0$-dimensional diagram of the HKS, using the \textbf{critical set method}.
             Learning rate is $0.2$, without momentum.
             The color encodes the time step.}
     \label{fig:hks_vineyard_crit}
\end{figure}

\begin{figure}[ht]
    \centering
    \includegraphics[]{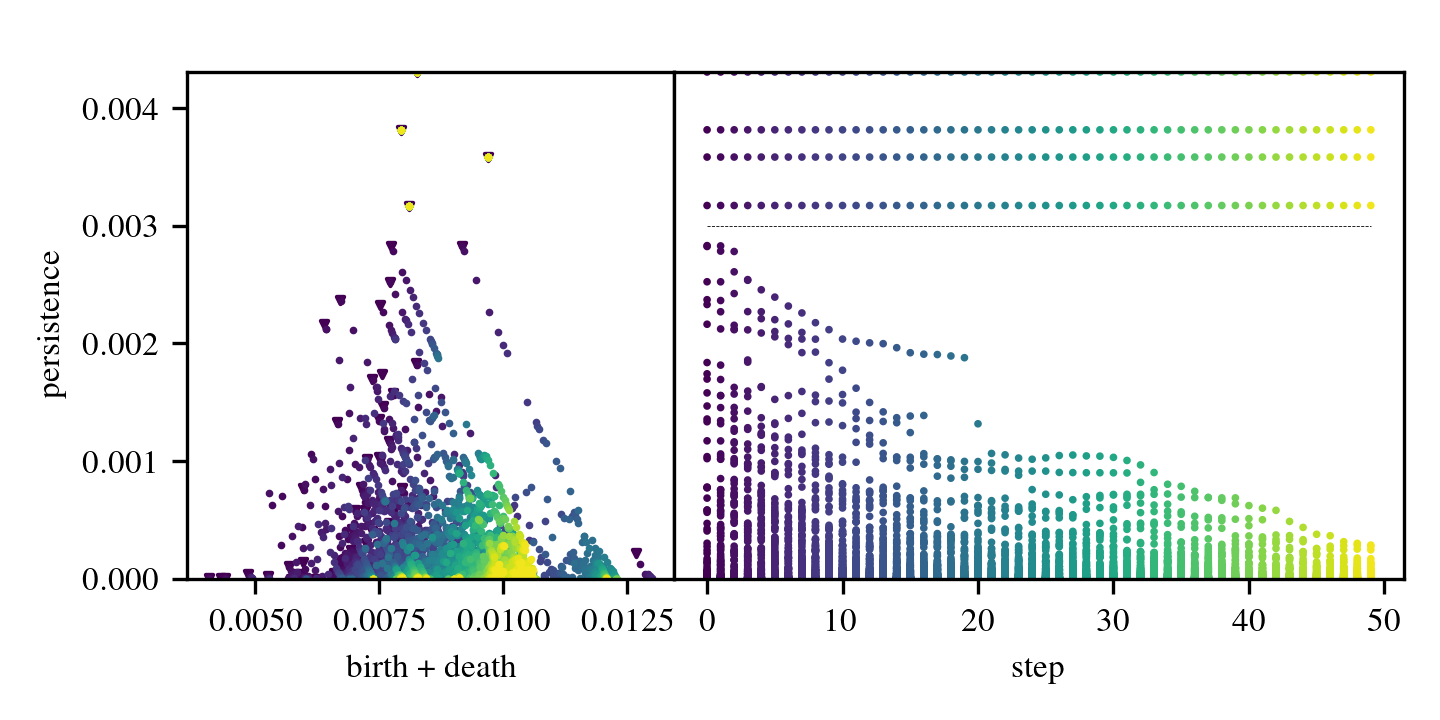}
    \vspace{-4ex}
    \caption{Vineyard of the optimization guided by the simplification of a
             superlevel set in a $0$-dimensional diagram of the HKS, using the \textbf{diagram method}.
             Learning rate is $0.2$, momentum $\gamma=0.5$.
             The color encodes the time step.}
     \label{fig:hks_vineyard_dgm}
\end{figure}

\begin{figure}[ht]
    \centering
    \includegraphics[]{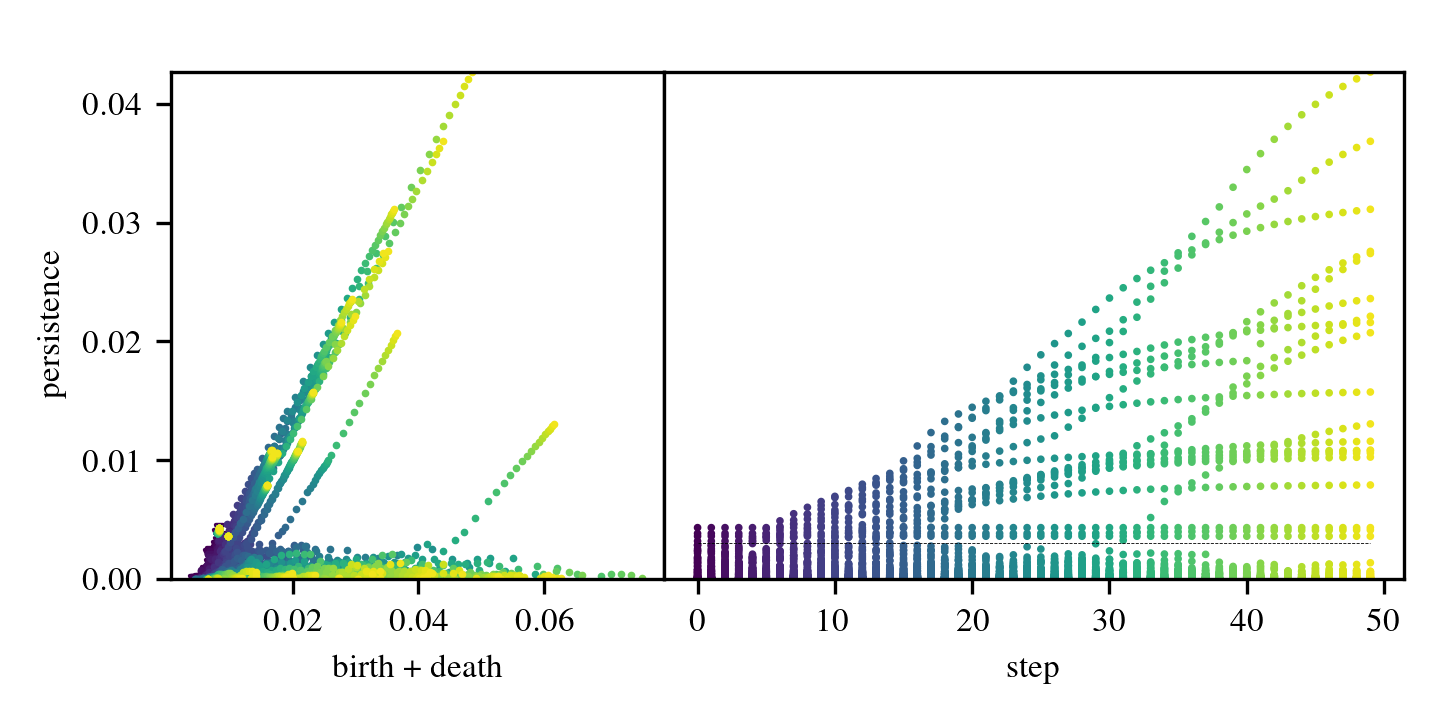}
    \vspace{-4ex}
    \caption{Vineyard of the optimization guided by the simplification of a
             superlevel set in a $0$-dimensional diagram of the HKS, using the \textbf{diagram method}.
             Learning rate is $0.2$, momentum $\gamma=0.9$.
             The color encodes the time step. Note that the most persistent points that we wanted to preserve are no longer fixed.}
     \label{fig:hks_vineyard_dgm_0.9}
\end{figure}

\paragraph{Direct optimization.} 
We performed an experiment similar to \cite{Poulenard2018}. We picked
a mesh from the SCAPE dataset \cite{anguelov2005scape} and computed HKS signature on it, using \cite{trailie_pyhks}.
We chose $t=0.2$ and $40$ eigenvectors.
Then we performed topological simplification, choosing $\eps$ to preserve the three most persistent
points in the zeroth diagram. The vineyards and diagrams are shown in \cref{fig:hks_vineyard_crit,fig:hks_vineyard_dgm}.
Here we find that the best performance of the diagram method was for a smaller value of $\gamma=0.5$,
with $\gamma=0.9$ the optimization diverges and starts moving points away from the diagonal, as we can see in \cref{fig:hks_vineyard_dgm_0.9}.
This highlights a disadvantage of the diagram loss: to perform well,
one needs to tune optimization parameters.
The critical set method with plain gradient descent quickly drives the loss to $0$,
while the diagram method does not achieve the same result even after $50$ steps, see
\cref{fig:hks_diagram_loss_comp}.

\begin{figure}[]
    \centering
    \includegraphics[]{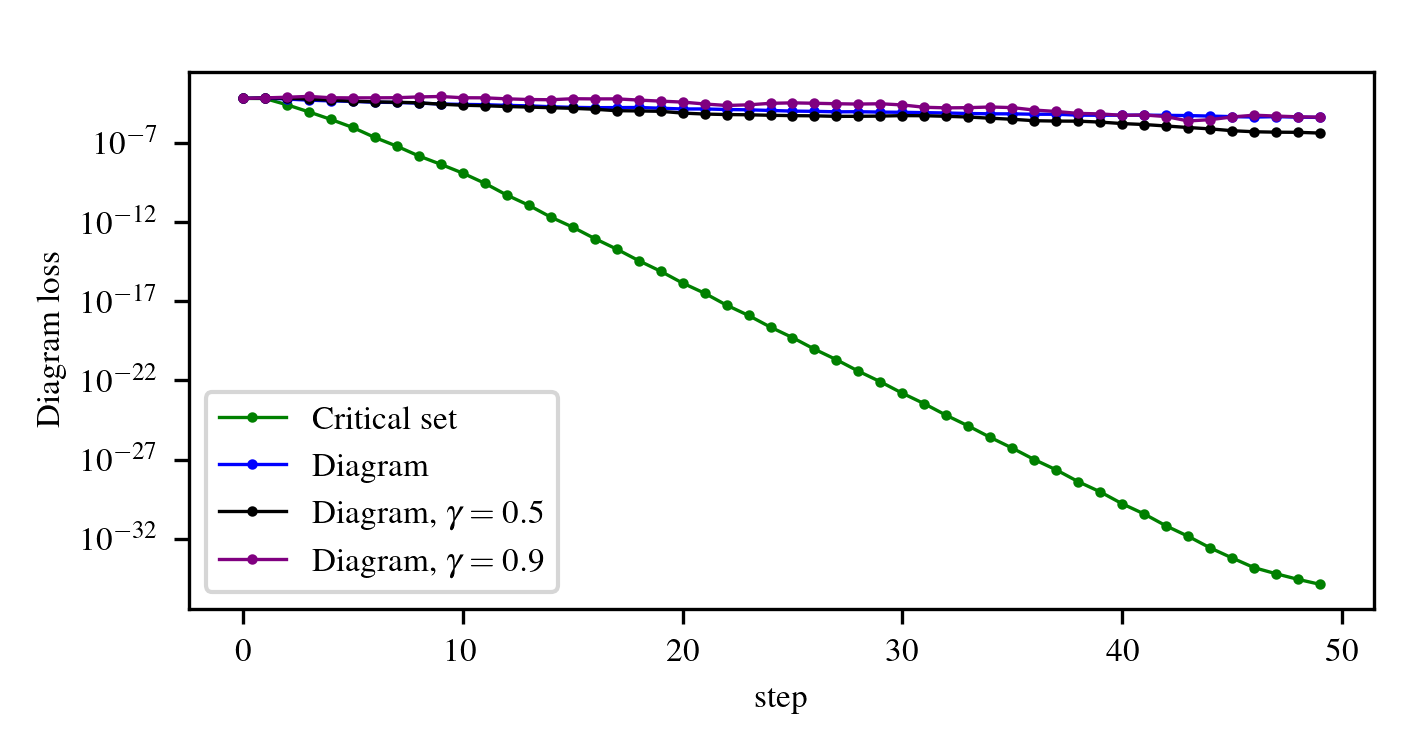}
    \vspace{-4ex}
    \caption{Comparison of the diagram losses during simplification of the HKS function.
             Diagram methods benefits from momentum, but the critical set significantly outperforms it.}
     \label{fig:hks_diagram_loss_comp}
\end{figure}

\cref{fig:hks_result_human} shows the results of the simplification.
It is hard to tell the difference visually; one can see that both methods remove
topological features by creating similar paths.

The diagram method takes 8.14 seconds. $50$ steps of the critical set method
take about $21.57$ seconds.
However, a fair comparison would be to run the critical set method
until it drives the loss below the value achieved by the diagram method, which
happens after 7 steps, after
only $2.96$ seconds.

\clearpage

\paragraph{Functional map regularization.}
Let us recall the methodology of the functional map correspondence method, following the notation
used in \cite{Poulenard2018}. We are given two manifolds (triangular meshes) $\manM$ and $\manN$.

First, we choose a set of basis functions, $k_\manM$ and $k_\manN$, on each manifold.
The basis functions are the eigenfunctions of the corresponding Laplace--Beltrami operator, $L_\manM$ or $L_\manN$.
Then we compute $k_d$ descriptors on each manifold and expand them in the corresponding basis.
We stack the column vectors with the coordinates of each descriptor into two matrices, $\matrA$ and $\matrB$,
of size $k_\manM \times k_d$ and $k_\manN \times k_d$.

More precisely, the chosen basis functions do not span the whole space of functions on the manifold,
unless we decide to use all of the eigenfunctions, because only in this case we have as many
functions in the basis as vertices. Accordingly, by expanding a function in the basis we actually
mean expanding its orthogonal projection on the subspace spanned by the first eigenfunctions.

The idea of a functional map correspondence is that instead of searching for a point-to-point correspondence $\manM \to \manN$,
we search for a linear mapping from the space of all real-valued functions on $\manM$ into the space of all functions on $\manN$.
Since we fix the bases, such a map is encoded by a matrix $\matrC$
of size $k_\manN \times k_\manM$. There are two reasonable requirements to impose on $\matrC$: 1) if the descriptors
are invariant under isometry, $\matrC$ must preserve them, i.e., $\matrC \matrA = \matrB$ and 2) the map should
commute with the Laplace--Beltrami operator. Our basis functions are eigenfunctions of the Laplace--Beltrami operator,
therefore we can express the second requirement as $\Lambda^\manN \matrC = \matrC \Lambda^\manM$,
where $\Lambda^\manM$ is the diagonal matrix whose diagonal consists of the first $k_\manM$ eigenvalues
of $L_\manM$, and $\Lambda^\manN$ is the diagonal matrices whose diagonal consists of the first $k_\manN$ eigenvalues
of $L_\manN$.

Thus, we obtain the first approximation of  $\matrC$ by solving the optimization problem
\begin{equation}
    \matrC = \arg \min_\matrX \| \matrX \matrA - \matrB \| + \eps \| \Lambda^\manN \matrX - \matrC \Lambda^\manM \|.
\label{eqn:matrC}
\end{equation}

In our experiments, we took two meshes form the SCAPE dataset.
We choose $k_\manM = k_\manN = k_d = 80$ and perform L-BFGS to solve \cref{eqn:matrC}.
The descriptors we chose are HKS function evaluated at different time values.

The topology comes into play in the second phase of the process.
While every bijective continuous mapping $f \colon \manM \to \manN$ gives rise
to the corresponding invertible linear map between functional spaces via pullback ($\phi \colon \manM \to \RR$
maps to $\phi \circ f^{-1} \colon \manN \to \RR$), the converse is not true.
Let us take $r$ connected regions on $\manM$ and let $\Omega_r$ be the indicator function of the union of the regions.
We slightly abuse the notation by writing $\matrC(\Omega_r)$ for the corresponding function $\manN \to \RR$ (first, $\Omega_r$ needs
to be projected onto the corresponding subspace).  The authors of \cite{Poulenard2018} show that it is reasonable to require the following:
the $0$-dimensional persistence diagram of $\matrC(\Omega_r)$ has exactly as
many points as the diagram of $\Omega_r$.
In other words, we should simplify the diagram of $\matrC(\Omega_r)$ to remove all but the first $r-1$
most persistent finite points (we assume $\manN$ and $\manM$ to be connected, so there is exactly one point at infinity).

We sample $r = 1$ random point and take all points
of the mesh that are at most $2$ hops away
as our region.
We want to optimize $\matrC$ to eliminate all finite points in the diagram 
of the image of the indicator function $\matrC(\Omega_1)$.
Crucially, unlike the rest of the examples in the paper, the optimization
parameters are the entries of matrix $\matrC$.

The behavior of the diagram loss is shown in \cref{fig:hks_corr_diagram_loss_comp}.
The advantage of the critical set method is evident, it rapidly drives the loss to $0$.

\begin{figure}[]
    \centering
    \includegraphics[]{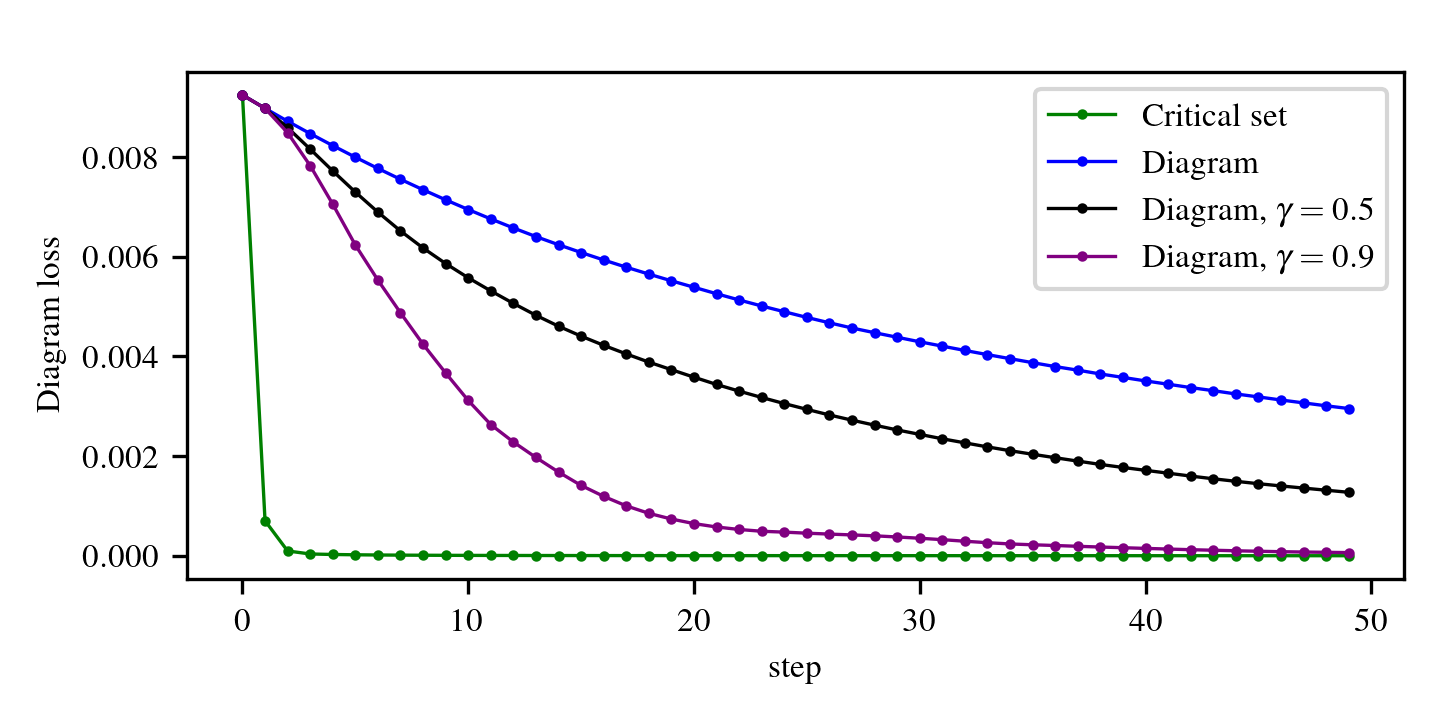}
    \vspace{-4ex}
    \caption{Comparison of the diagram losses during regularization of the functional map $\matrC$.
             Diagram methods benefits from momentum, but the critical set significantly outperforms it.}
     \label{fig:hks_corr_diagram_loss_comp}
\end{figure}

The vineyards are in \cref{fig:hks_corr_vineyard_crit,fig:hks_corr_vineyard_dgm,fig:hks_corr_vineyard_dgm_0.5,fig:hks_corr_vineyard_dgm_0.9}.

We should mention that these results are for simplification method that pushes the point $(b, d)$
towards $(d, d)$, i.e., increases the birth values.

\begin{figure}[ht]
    \centering
    \includegraphics[]{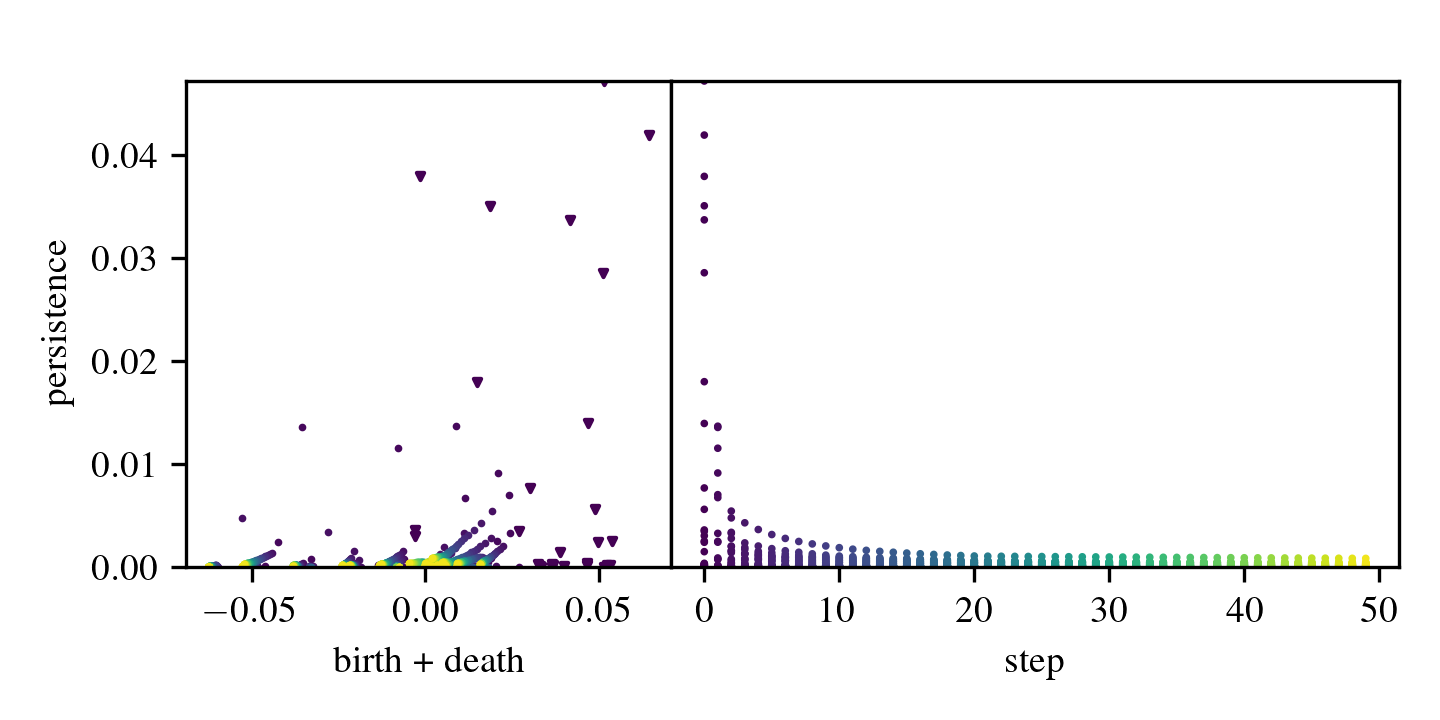}
    \vspace{-4ex}
    \caption{Vineyard of the optimization of the functional map $\matrC$ guided by the simplification of a
             superlevel set in a $0$-dimensional diagram of $\matrC(\Omega_1)$, using the \textbf{critical set method}.
             Learning rate is $0.2$, without momentum.
             The color encodes the time step.}
     \label{fig:hks_corr_vineyard_crit}
\end{figure}

\begin{figure}[]
    \centering
    \includegraphics[]{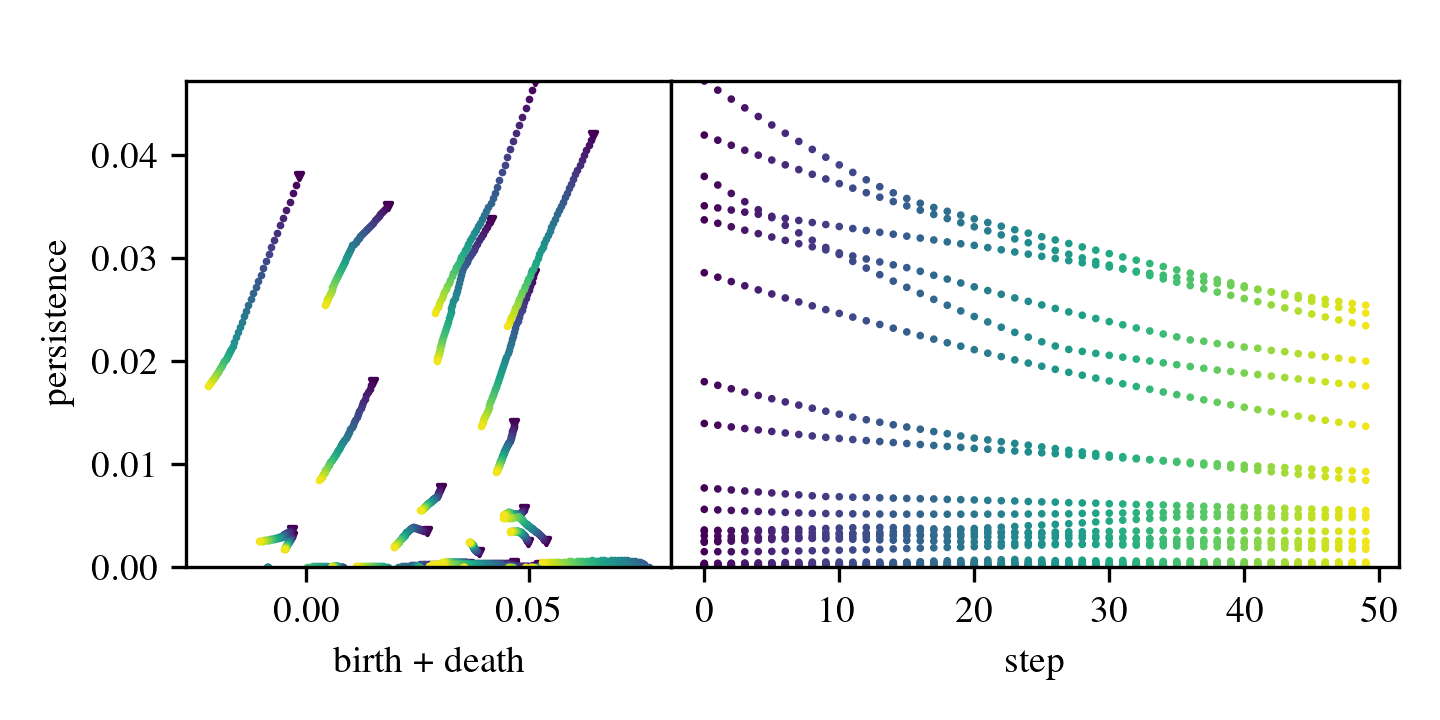}
    \vspace{-4ex}
    \caption{Vineyard of the optimization of the functional map $\matrC$ guided by the simplification of a
             superlevel set in a $0$-dimensional diagram of $\matrC(\Omega_1)$, using the \textbf{diagram method}.
             Learning rate is $0.2$, no momentum.
             The color encodes the time step.}
     \label{fig:hks_corr_vineyard_dgm}
\end{figure}

\begin{figure}[]
    \centering
    \includegraphics[]{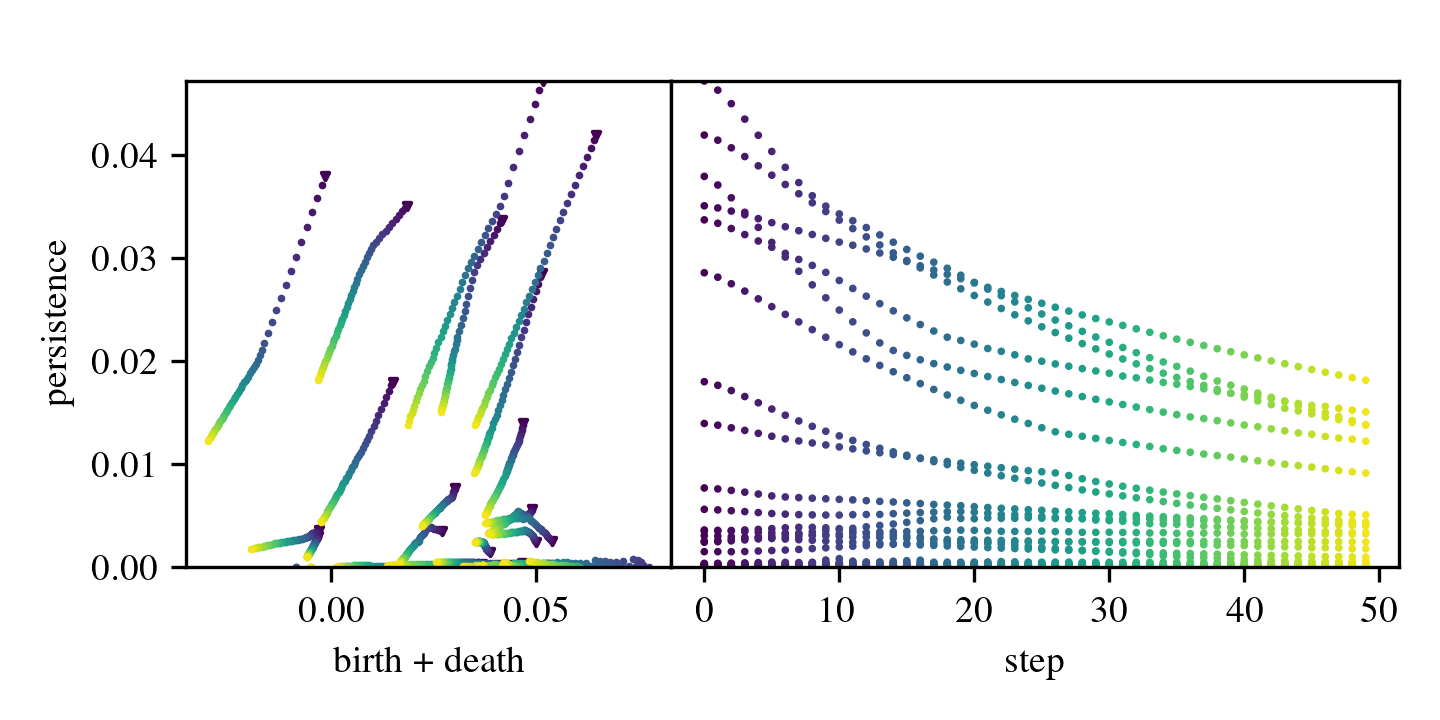}
    \vspace{-4ex}
    \caption{Vineyard of the optimization of the functional map $\matrC$ guided by the simplification of a
             superlevel set in a $0$-dimensional diagram of $\matrC(\Omega_1)$, using the \textbf{diagram method}.
             Learning rate is $0.2$, momentum $\gamma=0.5$.
             The color encodes the time step.}
     \label{fig:hks_corr_vineyard_dgm_0.5}
\end{figure}

\begin{figure}[]
    \centering
    \includegraphics[]{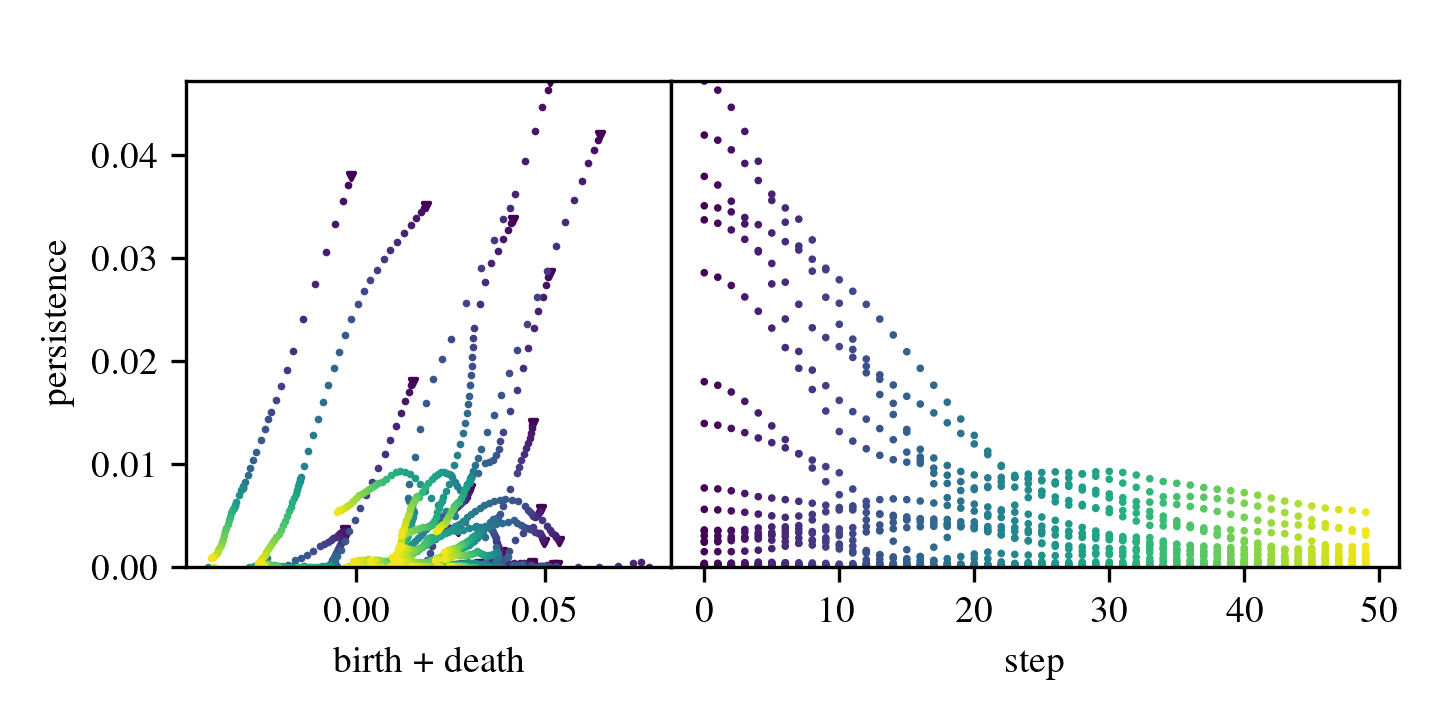}
    \vspace{-4ex}
    \caption{Vineyard of the optimization of the functional map $\matrC$ guided by the simplification of a
             superlevel set in a $0$-dimensional diagram of $\matrC(\Omega_1)$, using the \textbf{diagram method}.
             Learning rate is $0.2$, momentum $\gamma=0.9$.
             The color encodes the time step.}
     \label{fig:hks_corr_vineyard_dgm_0.9}
\end{figure}

\clearpage
\bibliographystyle{acm}
\bibliography{references}

\begin{thebibliography}{10}

\bibitem{anguelov2005scape}
{\sc Anguelov, D., Srinivasan, P., Koller, D., Thrun, S., Rodgers, J., and
  Davis, J.}
\newblock Scape: shape completion and animation of people.
\newblock In {\em ACM SIGGRAPH 2005 Papers}. 2005, pp.~408--416.

\bibitem{attali2009persistence}
{\sc Attali, D., Glisse, M., Hornus, S., Lazarus, F., and Morozov, D.}
\newblock Persistence-sensitive simplication of functions on surfaces in linear
  time.
\newblock In {\em TopoInVis' 09\/} (2009).

\bibitem{ripser}
{\sc Bauer, U.}
\newblock Ripser: efficient computation of {Vietoris--Rips} persistence
  barcodes.
\newblock {\em Journal of Applied and Computational Topology\/} (2021).

\bibitem{phat}
{\sc Bauer, U., Kerber, M., Reininghaus, J., and Wagner, H.}
\newblock Phat--persistent homology algorithms toolbox.
\newblock {\em Journal of symbolic computation 78\/} (2017), 76--90.

\bibitem{Bauer2012}
{\sc Bauer, U., Lange, C., and Wardetzky, M.}
\newblock Optimal topological simplification of discrete functions on surfaces.
\newblock {\em Discrete \& computational geometry 47}, 2 (2012), 347--377.

\bibitem{well-groups-1d}
{\sc Bendich, P., Edelsbrunner, H., Morozov, D., and Patel, A.}
\newblock Homology and robustness of level and interlevel sets.
\newblock {\em Homology, Homotopy and Applications 15}, 1 (2013), 51--72.

\bibitem{gabrielsson2020}
{\sc Br{\"u}el-Gabrielsson, R., Nelson, B.~J., Dwaraknath, A., Skraba, P.,
  Guibas, L.~J., and Carlsson, G.}
\newblock A topology layer for machine learning.
\newblock In {\em Proceedings of the International Conference on Artificial
  Intelligence and Statistics (AISTATS)\/} (2020), pp.~1553--1563.

\bibitem{Carriere2021}
{\sc Carriere, M., Chazal, F., Glisse, M., Ike, Y., Kannan, H., and Umeda, Y.}
\newblock Optimizing persistent homology based functions.
\newblock In {\em Proceedings of the 38th International Conference on Machine
  Learning\/} (2021), M.~Meila and T.~Zhang, Eds., vol.~139 of {\em Proceedings
  of Machine Learning Research}, PMLR, pp.~1294--1303.

\bibitem{clearing}
{\sc Chen, C., and Kerber, M.}
\newblock Persistent homology computation with a twist.
\newblock In {\em Proceedings 27th European Workshop on Computational
  Geometry\/} (2011), vol.~11, pp.~197--200.

\bibitem{chen2019}
{\sc Chen, C., Ni, X., Bai, Q., and Wang, Y.}
\newblock A topological regularizer for classifiers via persistent homology.
\newblock In {\em Proceedings of the International Conference on Artificial
  Intelligence and Statistics (AISTATS)\/} (2019), pp.~2573--2582.

\bibitem{vineyards}
{\sc Cohen-Steiner, D., Edelsbrunner, H., and Morozov, D.}
\newblock Vines and vineyards by updating persistence in linear time.
\newblock In {\em Proceedings of the Annual Symposium on Computational
  Geometry\/} (2006), pp.~119--126.

\bibitem{DDKL20}
{\sc Davis, D., Drusvyatskiy, D., Kakade, S., and Lee, J.~D.}
\newblock Stochastic subgradient method converges on tame functions.
\newblock {\em Foundations of computational mathematics 20}, 1 (Feb. 2020),
  119--154.

\bibitem{dualities}
{\sc de~Silva, V., Morozov, D., and Vejdemo-Johansson, M.}
\newblock Dualities in persistent (co)homology.
\newblock {\em Inverse problems 27}, 12 (Nov. 2011), 124003.

\bibitem{comp_top_book}
{\sc Edelsbrunner, H., and Harer, J.}
\newblock {\em Computational topology: an introduction}.
\newblock American Mathematical Society, 2010.

\bibitem{ph-survey}
{\sc Edelsbrunner, H., and Morozov, D.}
\newblock Persistent homology.
\newblock In {\em Handbook of Discrete and Computational Geometry}. Chapman and
  Hall/CRC, 2017, pp.~637--661.

\bibitem{simplification_2manifolds}
{\sc Edelsbrunner, H., Morozov, D., and Pascucci, V.}
\newblock Persistence-sensitive simplification functions on 2-manifolds.
\newblock In {\em Proceedings of the Annual Symposium on Computational
  Geometry\/} (2006), ACM, pp.~127--134.

\bibitem{well-groups}
{\sc Edelsbrunner, H., Morozov, D., and Patel, A.}
\newblock Quantifying transversality by measuring the robustness of
  intersections.
\newblock {\em Foundations of Computational Mathematics 11}, 3 (June 2011),
  345--361.

\bibitem{Gameiro2016}
{\sc Gameiro, M., Hiraoka, Y., and Obayashi, I.}
\newblock Continuation of point clouds via persistence diagrams.
\newblock {\em Physica D. Nonlinear phenomena 334\/} (Nov. 2016), 118--132.

\bibitem{magnetic_reconnection}
{\sc Guo, F., Li, H., Daughton, W., and Liu, Y.-H.}
\newblock Formation of hard power laws in the energetic particle spectra
  resulting from relativistic magnetic reconnection.
\newblock {\em Phys. Rev. Lett. 113\/} (Oct. 2014), 155005.

\bibitem{ELZ02}
{\sc {H.\ Edelsbrunner}, {D.\ Letscher}, and {A.\ Zomorodian}}.
\newblock Topological persistence and simplification.
\newblock {\em Discrete \& computational geometry 28}, 4 (Nov. 2002), 511--533.

\bibitem{osvd}
{\sc Klacansky, P.}
\newblock Open scientific visualization datasets.
\newblock \url{klacansky.com/open-scivis-datasets/}.

\bibitem{leygonie2021gradient}
{\sc Leygonie, J., Carri{\`e}re, M., Lacombe, T., and Oudot, S.}
\newblock A gradient sampling algorithm for stratified maps with applications
  to topological data analysis.
\newblock {\em arXiv:2109.00530\/} (2021).

\bibitem{LN21}
{\sc Luo, Y., and Nelson, B.~J.}
\newblock Accelerating iterated persistent homology computations with warm
  starts.
\newblock {\em arXiv:2108.05022\/} (2021).

\bibitem{Morozov2008}
{\sc Morozov, D.}
\newblock {\em Homological illusions of persistence and stability}.
\newblock PhD thesis, Duke University, 2008.

\bibitem{persistence-sensitive-optimization}
{\sc Nigmetov, A., Krishnapriyan, A.~S., Sanderson, N., and Morozov, D.}
\newblock Topological regularization via {Persistence-Sensitive} optimization.
\newblock {\em arXiv:2011.05290\/} (Nov. 2020).

\bibitem{Poulenard2018}
{\sc Poulenard, A., Skraba, P., and Ovsjanikov, M.}
\newblock Topological function optimization for continuous shape matching.
\newblock {\em Computer graphics forum: journal of the European Association for
  Computer Graphics 37}, 5 (Aug. 2018), 13--25.

\bibitem{rotstrat}
{\sc Rosenberg, D., Pouquet, A., Marino, R., and Mininni, P.~D.}
\newblock Evidence for {Bolgiano-Obukhov} scaling in rotating stratified
  turbulence using high-resolution direct numerical simulations.
\newblock {\em Physics of fluids 27}, 5 (May 2015), 055105.

\bibitem{solomon2021fast}
{\sc Solomon, Y., Wagner, A., and Bendich, P.}
\newblock A fast and robust method for global topological functional
  optimization.
\newblock In {\em International Conference on Artificial Intelligence and
  Statistics\/} (2021), PMLR, pp.~109--117.

\bibitem{Tierny2012}
{\sc Tierny, J., and Pascucci, V.}
\newblock Generalized topological simplification of scalar fields on surfaces.
\newblock {\em IEEE transactions on visualization and computer graphics 18}, 12
  (Dec. 2012), 2005--2013.

\bibitem{trailie_pyhks}
{\sc Trailie, C.}
\newblock Pyhks, 2018.
\newblock \url{github.com/ctralie/pyhks}.

\end{thebibliography}

\end{document}